\documentclass[manuscript,nonacm,screen]{acmart}
\setcopyright{none}

\usepackage{amsmath,amsthm}
\usepackage{shortcuts}
\usepackage{algorithm}
\usepackage{algcompatible}

\theoremstyle{definition}
\newtheorem{remark}{Remark}
\newtheorem{assumption}{Assumption}

\usepackage[capitalise]{cleveref}
\Crefname{assumption}{Assumption}{Assumptions}
\Crefname{section}{Section}{Sections}

\newcommand{\X}{\mathcal X}
\newcommand{\Yobs}{Y}
\newcommand{\Ypot}{Y^*}

\newcommand{\cvar}{\op{CVaR}}
\newcommand{\cvarat}[1]{\cvar_{#1}}

\newcommand{\ate}{\bar\tau}
\newcommand{\ite}{\delta}
\newcommand{\catef}{\tau}
\newcommand{\cate}{\catef(X)}

\newcommand{\ostar}{\textcircled{$\star$}}

\newcommand{\tmu}{\check\mu}
\newcommand{\tprop}{\check e}
\newcommand{\tcate}{\check\catef}
\newcommand{\tbeta}{\check\beta}

\newcommand{\smu}{\tilde\mu}
\newcommand{\sprop}{\tilde e}
\newcommand{\scate}{\tilde\catef}
\newcommand{\sbeta}{\tilde\beta}

\let\citealp\citep
\begin{document}

\title{Treatment Effect Risk: Bounds and Inference}

\author{Nathan Kallus}
\email{kallus@cornell.edu}
\affiliation{%
  \institution{Netflix and Cornell University}
}

\renewcommand{\shortauthors}{Nathan Kallus}

\begin{abstract}Since the average treatment effect (ATE) measures the change in social welfare, even if positive, there is a risk of negative effect on, say, some 10\% of the population. Assessing such risk is difficult, however, because any one individual treatment effect (ITE) is never observed, so the 10\% worst-affected cannot be identified, while distributional treatment effects only compare the first deciles within each treatment group, which does not correspond to any 10\%-subpopulation. In this paper we consider how to nonetheless assess this important risk measure, formalized as the conditional value at risk (CVaR) of the ITE-distribution. We leverage the availability of pre-treatment covariates and characterize the tightest-possible upper and lower bounds on ITE-CVaR given by the covariate-conditional average treatment effect (CATE) function. We then proceed to study how to estimate these bounds efficiently from data and construct confidence intervals. This is challenging even in randomized experiments as it requires understanding the distribution of the unknown CATE function, which can be very complex if we use rich covariates so as to best control for heterogeneity. We develop a debiasing method that overcomes this and prove it enjoys favorable statistical properties even when CATE and other nuisances are estimated by black-box machine learning or even inconsistently. Studying a hypothetical change to French job-search counseling services, our bounds and inference demonstrate a small social benefit entails a negative impact on a substantial subpopulation.\end{abstract}

\maketitle

\section{Introduction}\label{sec:intro}

Policymakers and project managers regularly conduct randomized experiments (``A/B tests'') to assess potential changes to policy or product. A key metric is the \emph{average treatment effect (ATE)}, the difference in the population-average outcome when everyone or no one is treated. 
ATEs are easily estimated by differences in the sample-average outcome within treatment groups, barring interference. 
Estimation from observational data is also possible under appropriate assumptions, \eg, unconfoundedness \citep{imbens2015causal}.
Identifying an individual's outcome with their utility -- as we will throughout this paper --
the ATE is the difference in social welfare in these two counterfactual scenarios. 
By linearity, this coincides with the population-average of each \emph{individual}'s treatment effect, the difference in their own utility in the two counterfactual scenarios. 

It is widely recognized, however, that treatment effects can vary widely between individuals \citep{heckman1997making,crump2008nonparametric}. Thus, even if the ATE is positive, there is a \emph{risk} that many individuals are harmed by the proposed change. Crucially, \emph{distributional} treatment effects (DTEs), which compare the two counterfactual utility distributions beyond their means, \emph{cannot} capture this risk. Indeed, \citet{imbens2009recent} note ``quantile effects are defined as differences between quantiles of the two marginal potential outcome distributions, and not as quantiles of the unit level effect.'' They nonetheless advocate for the former because policy ``choice should be governed by preferences of the policymaker over these distributions.'' However, such rational-decision-making framing presumes a policymaker facing a choice between lotteries drawing at random from individual outcomes. Instead, concerned with equity beyond social welfare, we should worry about the individuals, not the policymaker. Hypothetically, harm to some individuals is possible even when the ``treat-all'' utility distribution first-order-dominates ``treat-none'' so that \emph{any} expected-increasing-utility-function-maximizer would choose ``treat-all.''

One way to gain further insight into heterogeneity and hence inequities is to consider conditional ATEs (CATEs) given pre-treatment covariates.
For example, if we observe a discrete sensitive attribute (\eg, race), we can simply compare the CATE in each attribute-value group.\footnote{We may still make some inferences on these even if we do not observe such attributes; see \citealp{chen2019fairness,kallus2021assessing}.} But it may not always be clear what are relevant such attributes and whether we are omitting important ones.
Given rich and continuous covariates, we can still reliably learn the CATE function by leveraging recent advances in causal machine learning \citep{slearner,xlearner,drlearner,rlearner,causaltree,causalforest}. It may still not be clear, nonetheless, whether the covariates are relevant for fairness considerations, what groups are captured in this way, and/or how to summarize the many individual predictions of complex machine-learned CATEs.

It is therefore particularly appealing to focus directly on the distribution of \emph{individual} treatment effects (ITEs), such as the average effects among the worst-affected 10\%, 20\%, \etc, corresponding to the conditional value at risk (CVaR) of this distribution.
The challenge is that no ITE can ever be observed -- the so-called Fundamental Problem of Causal Inference. Nonetheless, regardless of whether covariates are meaningful for fairness considerations, if they control for heterogeneity, CATE may predict ITE well.
%
%
%
In this paper, we leverage this to proxy these important but unidentifiable treatment-effect risk measures. Specifically, we provide the tightest-possible upper and lower bounds given by CATE on the CVaR of ITE. By construction these are functions of distributions of observables. What remains is inference from data, whether experimental or observational. Since the CATE function can be high-dimensional, especially if we use a lot of covariates to control for heterogeneity, inference is difficult and na\"ive plug-in approaches fail. We design debiased estimators and confidence intervals for our bounds that overcome this challenge by being exceedingly robust: given rough, machine-learned estimates of CATE and other nuisances, they behave as though we used perfect estimates; they remain consistent even when some nuisances are mis-estimated; and surprisingly they remain valid as bounds even when CATE is mis-estimated.
We conclude by using our tools to illustrate treatment-effect risk in a case study of job-search-assistance benefits.

\section{Problem Set Up and Definitions}

Each individual in the population is associated with two potential outcomes, $\Ypot(0),\,\Ypot(1)\in\Rl$, corresponding to individual utility under ``treat-all'' and ``treat-none,'' respectively, and baseline covariates (observable characteristics), $X\in\X$.
The ITE, ATE, and CATE are, respectively,
\begin{align*}
\ite&=\Ypot(1)-\Ypot(0),
\qquad\ate=\E[\Ypot(1)]-\E[\Ypot(0)]=\E\ite=\E\cate\\
\cate&=\E[\delta\mid X]=\mu(X,1)-\mu(X,0),\quad\text{where $\mu(X,a)=\E[\Ypot(a)\mid X]$}.
\end{align*}
%
We assume $\E\ite^2<\infty$ throughout.

Of interest is the average effect among the $(100\times\alpha)\%$-worst affected, formalized by $\cvarat\alpha(\ite)$, where for any $Z$ \citep{rockafellar2000optimization}\footnote{CVaR is sometimes defined for the right tail, corresponding to our $-\cvarat\alpha(-Z)$.}
\begin{equation}\label{eq:cvar}\cvarat\alpha(Z)=\sup_\beta\prns{\beta+\frac1\alpha\E(Z-\beta)_-},\end{equation}
where $(u)_-=u\wedge0$.
The $\sup$ is attained by $\beta$ equal the $\alpha$-quantile:
\begin{equation}\label{eq:quantile}
F_Z^{-1}(\alpha)=\inf\fbraces{\beta:F_Z(\beta)\geq\alpha},
\quad\text{where}~F_Z(z)=\Prb{Z\leq z}.
\end{equation}

Provided $F_Z(F_Z^{-1}(\alpha))=\alpha$ (\eg, $Z$ continuous), then $\cvarat\alpha(Z)=\Eb{Z\mid Z\leq F_Z^{-1}(\alpha)}$. Otherwise, $\cvarat\alpha(Z)\in[\Eb{Z\mid Z< F_Z^{-1}(\alpha)},\,\Eb{Z\mid Z\leq F_Z^{-1}(\alpha)}]$, and, unlike these two endpoints, $\cvarat\alpha(Z)$ is continuous in $\alpha$ and coherent \citep{artzner1999coherent}. It is therefore the \emph{correct} generalization of ``average of the $(100\times\alpha)\%$-lowest values'' when ambiguous due to discontinuities.


We consider data from a randomized experiment or observational study. Each individual is associated with a treatment $A\in\{0,1\}$, and we observe the \emph{factual} outcome $\Yobs=\Ypot(A)$ (never $\Ypot(1-A)$). The data is $(X_i,A_i,\Yobs_i)\sim(X,A,Y)$, $1\leq i\leq n$.
We assume unconfoundedness throughout: $\Ypot(a)\indep A\mid X$.\footnote{And $\Yobs=\Ypot(A)$ assumes non-interference \citep{rubin1986comment}.}
Randomized experiments (our focus) ensure this by design (often with $X\indep A$). 
Our results nonetheless extend to observational settings assuming unconfoundedness.
Under unconfoundedness, ATE and CATE are identifiable, \ie, are functions of the $(X,A,\Yobs)$-distribution: $\mu(X,a)=\E[\Yobs\mid X,A=a]$, $\cate=\mu(X,1)-\mu(X,0)$, $\ate=\E\cate$ ($=\E[\Yobs\mid A=1]-\E[\Yobs\mid A=0]$ if $X\indep A$).
Define also
the propensity score $e(X)=\Prb{A=1\mid X}$ and marginal-outcome regression $\bar\mu(X)=\E[\Yobs\mid X]=e(X)\mu(X,1)+(1-e(X))\mu(X,0)$.

We now illustrate treatment-effect risk and its \emph{un}identifiability, which motivates us to consider the tightest-possible \emph{identifiable} bounds (\cref{sec:bounds}) and inference thereon (\cref{sec:inference}).

\begin{example}[Simple Example]\label{ex:simple}
Suppose
$$\begin{pmatrix}\Ypot(0)\\\Ypot(1)\end{pmatrix}\sim\mathcal N\prns{\begin{pmatrix}\mu(0)\\\mu(1)\end{pmatrix},\,\begin{pmatrix}1&\rho\\\rho&1\end{pmatrix}},~\mu(1)\geq\mu(0),~\rho\in[-1,1].$$
If $\ate=\mu(1)-\mu(0)>0$, the $\Ypot(1)$-distribution \emph{first-order-dominates} $\Ypot(0)$. If $\mu(1)=\mu(0)$, the distributions are \emph{indistinguishable}.
However, the ITE-distribution depends on $\rho$: $\ite\sim\mathcal N(\mu(1)-\mu(0),\sqrt{2-2\rho})$, 
$\cvarat{0.1}(\ite)=\ate-1.75 \sqrt{2-2\rho}$. 
The unidentifiability of $\cvarat{0.1}(\ite)$ follows because the $(A,\Yobs)$-distribution is fixed given just $\mu(0),\mu(1),\Prb{A=1}$ while $\cvarat{0.1}(\ite)$ varies with $\rho$. 
\end{example}

\begin{remark}[Covariate-conditional policies]
Treat (\ie, rollout to) all or none is often the choice faced by project managers, but given covariates we can learn covariate-conditional treatment policies \citep{kallus2018balanced,athey2017efficient,qian2011performance,zhao2012estimating,kallus2021minimax,
kitagawa2018should}. Learning aside, treating only when $\cate>0$ ensures all covariate-defined groups have nonnegative group-average effects.\footnote{However, even this ideal can induce disparate impacts \citep{kallus2019assessing}.}
Personalizing on all available covariates
is however generally infeasible due to 
operational, non-stationarity, and/or ethical/reputational concerns.
Nonetheless, given any policy $\pi:\X\to\{0,1\}$, we may simply redefine ITE as $Y(\pi(X))-Y(0)$ and our results still apply. This is especially relevant when $\pi$ personalizes on some covariates and the rest explain heterogeneity conditionally thereon.
\end{remark}

\begin{remark}[Risk of observed vs unobserved variables]
CVaR is an example of coherent risk measures \citep{artzner1999coherent}, which are used to assess distributions beyond expectations and are equivalent to distributionally-robust worst-case expectations \citep{ruszczynski2006optimization}. For example, 
CVaR is the worst-case expectation among distributions with Radon-Nikodym derivative to the given distribution bounded by $1/\alpha$.
Other distributional divergences can also define ambiguity sets \citep[\eg,][]{ben2013robust,bertsimas2018robust,esfahani2018data}.
Alternative approaches limit the \emph{complexity} of subpopulations \citep{NEURIPS2020_07fc15c9,kearns2018preventing}.
In both finance \citep{krokhmal2002portfolio}, distributionally-robust supervised learning \citep{bagnell2005robust}, demographics-free fair learning \citep{NEURIPS2020_07fc15c9},
and CVaR-DTEs \citep{kallus2019localized},
the variable whose risk is of interest is \emph{always observed}. \Eg,
model loss on each training example is observed.
In contrast, we consider risk of an \emph{unobserved variable}, hence we study bounds in \cref{sec:bounds}. 
For inference, we are uniquely concerned with risk of an \emph{unknown function}, hence we develop learning-robust methods in \cref{sec:inference}.
\end{remark}


\section{Bounds}\label{sec:bounds}


\subsection{Upper Bound: The CATE-CVaR}

An upper bound on $\cvarat\alpha(\ite)$ is crucial: if negative or substantially below ATE, the change poses certifiable risk or inequity to an $(100\times\alpha)\%$-subpopulation.

\begin{theorem}[Upper Bound by CATE-CVaR]\label{thm:cvarbound}
\begin{equation}\label{eq:cvarbound}
\cvarat\alpha(\ite)\leq\cvarat\alpha(\cate).
\end{equation}
Moreover, given any $X$-distribution and integrable $\catef:\X\to\Rl$, some $(X,\ite)$-distribution
has the given $X$-marginal, $\cate=\E[\ite\mid X]$, and \cref{eq:cvarbound} holding with equality.
\end{theorem}

Since $\cate$ represents our \emph{best guess} for $\ite$ (in squared error), imputing the unknown $\ite$ with $\cate$ seems reasonable.
\Cref{thm:cvarbound} shows this in fact provides an upper bound.\footnote{\label{footnote:coherentrisk}\Cref{eq:cvarbound} extends to any coherent risk by writing $\ite=\cate+(\ite-\cate)$ and using sub-additivity.}
If $\cate$ is continuous, $\cvarat\alpha(\cate)=\Efb{\delta\mid \cate\leq F_{\cate}^{-1}(\alpha)}$,
and \cref{eq:cvarbound} is intuitive: 
$\cvarat\alpha(\ite)$ is worst average effect among \emph{all} $(100\times\alpha)\%$-subpopulations, while
$\cvarat\alpha(\cate)$ only among $X$-defined subpopulations.
This bound is also tight: given just $\cate$, it cannot be improved.\footnote{The bound need not be tight given the $(X,A,\Yobs)$-distribution, which characterizes more than the mean of the $(\delta\mid X)$-distribution, as described by the Fr\'echet-Hoeffding bounds. We focus on best bounds given just by CATE, which is the common tool to understand effect heterogeneity in practice.}

\Cref{thm:cvarbound} implies an ordering:
\begin{equation}\label{eq:ordering}\cvarat{\alpha_1}(\ite)\leq\cvarat{\alpha_2}(\ite)\leq\cvarat{\alpha_2}(\cate)\leq\cvarat{\alpha_3}(\cate)\leq\ate~~~\forall~0<\alpha_1\leq\alpha_2\leq\alpha_3\leq1.\end{equation}

\begin{remark}[CVaR as summary of CATE]\label{remark:summary}
Aside from being a bound, $\cvarat\alpha(\cate)$ is of independent interest as a summary of effect heterogeneity along meaningful covariates $X$ of explicit equity concern. 
When $X$ is more than a few discrete groups, understanding the many facets of estimated heterogeneity is challenging, both interpretationally and statistically. We could test for $X$-heterogeneity \citep{crump2008nonparametric,sawilowsky1990nonparametric,gail1985testing,davison1992treatment}.\footnote{There are also tests for heterogeneity \emph{not} explained by $X$ \citep{ding2019decomposing,ding2016randomization}. These, like us, leverage bounds on unidentifiable quantities.} 
\Eg, omnibus test $H_0:0\in\argmin_\gamma\E(\cate-\ate-\gamma^\top (X-\E X))^2$ \citep{chernozhukov2018generic}.
This, however, may detect minor heterogeneity in small subpopulations, may not assess magnitude or direction, and may be inappropriate if we expect heterogeneity. In contrast, $\cvarat\alpha(\cate)$ is a simple, meaningful summary of $\cate$.
Inference, however, is a challenge. We tackle this in \cref{sec:inference}.
%
\end{remark}

\begin{remark}[Inter-quantile averages of CATE]\label{remark:gate}
CVaR of CATE can in fact permit us to summarize average effects in the middle, not just the tails. Consider any $0<\alpha<\alpha'<1$. Provided that $F_{\cate}(F_{\cate}^{-1}(\alpha))=\alpha$, $F_{\cate}(F_{\cate}^{-1}(\alpha'))=\alpha'$ (\eg, $\cate$ is continuous), we have that
\begin{equation}\label{eq:interquantile}
\Eb{\Ypot(1)-\Ypot(0)\mid F_{\cate}^{-1}(\alpha)<\cate\leq F_{\cate}^{-1}(\alpha')}
=
\frac{\alpha'\cvarat{\alpha'}(\cate)-\alpha\cvarat{\alpha}(\cate)}{\alpha'-\alpha}.
\end{equation}
\Cref{eq:interquantile} is the average effect among individuals with CATE between the $\alpha$- and $\alpha'$-quantiles. A similar but different quantity is considered in \citet{chernozhukov2018generic}: the average effect among individuals in inter-quantile ranges of an \emph{estimate} of CATE fit on a split sample, rather than the true CATE. They consider averaging this over splits, but that average still need not correspond to \cref{eq:interquantile}, and this approach is \emph{not} robust to errors in the CATE estimate, meaning these errors will propagate to non-negligible terms in the estimate and its variance. In contrast, by leveraging the unique optimization structure of CVaR, in \cref{sec:inference} we provide an estimator that \emph{is} robust to such errors, allowing us to estimate the CVaR of the true CATE, rather than a split-sample-estimated CATE.
By writing \cref{eq:interquantile} using CVaR, we can then leverage these results to get robust estimates for inter-quantile averages, as we will explain in \cref{remark:comparing}.
\end{remark}







\begin{remark}[\emph{Who} is negatively affected?]
Suppose we find $\cvarat\alpha(\cate)<0$ while $\ate>0$, where $\alpha$ is ``substantial'' -- the social-welfare benefit of the proposal is borne by some substantial negatively-impacted subpopulation.
While that may already cool enthusiasm for the proposal, we may wonder \emph{who} are the harmed individuals, \eg, to help design a new, better treatment.

Assuming continuity, $\cvarat\alpha(\cate)$ is the ATE among individuals with $\cate\leq F_{\cate}^{-1}(\alpha)$ -- an \emph{identifiable} group. A question is interpretation.
This is easy if $\cate$ is linear or tree (or estimated using such models, which still gives a bound per \cref{thm:doublyvalid}). 
We can also consider summaries of this group, \eg, fraction belonging to sensitive groups, or learn simpler models to explain membership \citep{lakkaraju2019faithful,ribeiro2016should}.
%
Alternatively, given we detect substantial inequities,
%
we can \emph{separately} investigate which variables negatively modulate treatment effect by, \eg, studying
$\argmin_\gamma\E(\cate-\ate-\gamma^\top X)^2$ \citep{drlearner,chernozhukov2018generic}.
\end{remark}

\subsection{Lower Bounds under Limited Residual Heterogeneity Range}\label{sec:lowerboundbounded}

Much as we try to best control for heterogeneity, disparate effect-predictiveness of covariates may mean some negative ITEs are averaged out and hidden while others are singled out. A remedy when concerned about disproportionate predictiveness among sensitive groups (\eg, race) would be to include these (or proxies) within $X$. But, we may always worry about missing something. 
A lower bound can provide assurances about what the upper bound may be missing.


This depends on how much residual heterogeneity remains.
Our first set of lower bounds limit the range of residual heterogeneity, \ie, almost-sure bounds on $\ite-\cate$, while our second set of lower bounds limit its variance, \ie, bounds on $\op{Var}(\ite\mid X)=\E(\ite-\cate)^2$.

\begin{theorem}\label{thm:cvarlb2}
Suppose $\abs{\cate-\ite}\leq b$. Then
\begin{equation}\label{eq:cvarlb2}
\cvarat\alpha(\ite)\geq\sup_\beta\prns{\beta+\frac1{2\alpha}\E[(\cate-b-\beta)_-]+\frac1{2\alpha}\E[(\cate+b-\beta)_-]}.
\end{equation}
Moreover, given any $X$-distribution and integrable $\catef:\X\to\Rl$, some $(X,\ite)$-distribution has the given $X$-marginal, $\cate=\E[\ite\mid X]$, $\abs{\cate-\ite}\leq b$, and \cref{eq:cvarlb2} holding with equality.
\end{theorem}

The right-hand side of \cref{eq:cvarlb2} is the $\alpha$-CVaR of the equal-mixture distribution of $\cate-b$ and $\cate+b$. It reduces to $\cvarat\alpha(\cate)$ when $b=0$ (equivalent to $\delta=\cate$). 
When $\alpha=1$, it becomes $\ate$ for any $b\geq0$ (as necessary for tightness).
The lower bound is established via weak semi-infinite duality and its tightness by exhibiting the equal-mixture distribution.

Since $(\cate\pm b-\beta)_-\geq (\cate-\beta)_--b$, \cref{eq:cvarlb2} upper bounds $\cvarat\alpha(\cate)-b$. This simpler bound is tight if we only assume a one-sided-bounded range.

\begin{theorem}\label{thm:cvarlb1}
Suppose $\cate-\ite\leq b$. Then
\begin{equation}\label{eq:cvarlb1}
\cvarat\alpha(\ite)\geq\cvarat\alpha(\cate)-b.
\end{equation}
Moreover, for $\alpha<1$, given any $\varepsilon>0$, $X$-distribution, and integrable $\catef:\X\to\Rl$, some $(X,\ite)$-distribution has the given $X$-marginal, $\cate=\E[\ite\mid X]$, $\cate-\ite\leq b$, and \cref{eq:cvarlb1} holding with equality up to $\varepsilon$-error.
\end{theorem}

The lower bound is immediate and its tightness given by exhibiting a skewed two-point-mass distribution.
For $\alpha=1$, \cref{eq:cvarlb1} simply reads $\ate\geq\ate-b$, but for \emph{any} $\alpha<1$, \cref{eq:cvarlb1} is actually \emph{tight}.



\subsection{Lower Bounds under Limited Residual Heterogeneity Variance}\label{sec:lowerboundvariance}

Limiting residual heterogeneity within a range may be implausible, or plausible only with large constants, yielding a weak bound. We next explore the implication of the residual ITE-variance after controlling for $X$, which we can bound given observables.

\begin{theorem}\label{thm:cvarlb3}
Suppose $\op{Var}(\ite\mid X)\leq \bar\sigma^2(X)$ for some integrable $\bar\sigma^2:\X\to\Rl_+$. Then
\begin{equation}\label{eq:cvarlb3}
\cvarat\alpha(\ite)\geq\sup_\beta\prns{\beta+\frac1{2\alpha}\Eb{\cate-\beta-\sqrt{(\cate-\beta)^2+\bar\sigma^2(X)}}}.
\end{equation}
Moreover, given any $\varepsilon>0$, $X$-distribution, and integrable $\catef:\X\to\Rl$, some $(X,\ite)$-distribution has the given $X$-marginal, $\cate=\E[\ite\mid X]$, $\op{Var}(\ite\mid X)\leq \bar\sigma^2(X)$, and \cref{eq:cvarlb3} holding with equality up to $\varepsilon$-error.
\end{theorem}

The proof of \cref{thm:cvarlb3} leverages strong duality for convex semi-infinite optimization.
Note 
\cref{eq:cvarlb3} equals $\cvarat\alpha(\cate)$ whenever $\bar\sigma^2(X)=0$ and $\ate$ whenever $\alpha=1$. 
Since $\abs{\ite-\cate}\leq b\implies\op{Var}(\ite\mid X)\leq b^2$, plugging $\bar\sigma^2(X)=b^2$ into \cref{eq:cvarlb3} must be looser than \cref{eq:cvarlb2} by tightness. Triangle inequality verifies this directly:
$\sum_{\pm}(\cate\pm b-\beta)_-=\cate-\beta-\frac12\sum_{\pm}\abs{\cate\pm b-\beta}\geq \cate-\beta-\sqrt{(\cate-\beta)^2+b^2}$.

A residual-variance bound is both more plausible and easier to calibrate than an absolute bound. Letting $\rho(X)=\op{Corr}(Y(0),Y(1)\mid X)\in[-1,1]$, we have
\begin{align}\label{eq:vardecomp}
\op{Var}(\ite\mid X)=\op{Var}(\Yobs\mid X,A=0)+\op{Var}(\Yobs\mid X,A=1)-2\rho(X)\op{Var}^{1/2}(\Yobs\mid X,A=0)\op{Var}^{1/2}(\Yobs\mid X,A=1),
\end{align}
where all terms but $\rho(X)$ are identifiable.
Thus, postulating different potential-outcome correlations, we obtain different bounds.
\Cref{eq:vardecomp} is maximized for $\rho(X)=-1$, which is tight, as all correlations are realizable. Thus, plugging $\bar\sigma^2(X)=(\op{Var}^{1/2}(\Yobs\mid X,A=0)+\op{Var}^{1/2}(\Yobs\mid X,A=1))^2$ into \cref{eq:cvarlb3} yields a tight lower bound on ITE-CVaR, given conditional expectations and variances. We may obtain better bounds if we postulate larger $\rho(X)$.

\Cref{thm:cvarlb3} also implies a simpler but looser bound.
\begin{corollary}\label{cor:cvarlb4}
\begin{align}
\label{eq:cvarlb4a}
0\leq\cvarat\alpha(\cate)-&\cvarat\alpha(\ite)\leq\frac{1}{2\alpha}\Eb{\op{Var}^{1/2}(\ite\mid X)}\\
\label{eq:cvarlb4b}
&\leq\frac{1}{2\alpha}\Eb{\op{Var}^{1/2}(\Yobs\mid X,A=0)+\op{Var}^{1/2}(\Yobs\mid X,A=1)}\\
\label{eq:cvarlb4c}
&
\leq\frac{1}{2\alpha}\sqrt{\Eb{(\Yobs-\mu(X,A))^2\mid A=0}}+\frac{1}{2\alpha}\sqrt{\Eb{(\Yobs-\mu(X,A))^2\mid A=1}}
.
\end{align}
\end{corollary}

\Cref{eq:cvarlb4a} more transparently bounds the slack in \cref{eq:cvarbound} in terms of residual effect variance. However, it is not tight, as can be seen for $\alpha=1$.
\Cref{eq:cvarlb4c} is even looser but appealing as it avoids
$\op{Var}(\Yobs\mid X,A)$, depending only on the root-mean-squared error of regressing $Y$ on $X$ for each $A\in\{0,1\}$ (\ie, the numerator of nonparametric $R^2$).

\section{Inference}\label{sec:inference}

We next turn to estimating the bounds developed in \cref{sec:bounds} and constructing confidence intervals.
Recall our data $(X_i,A_i,\Yobs_i)\sim(X,A,Y)$, $1\leq i\leq n$, may be experimental or observational. 
%
%
The only relevant technical difference between these two cases is whether propensity, $e(X)=\Prb{A=1\mid X}$, is known or not. 
While it matters not here, note that $e(X)$ is usually constant in experiments ($A\indep X$). In observational settings $e(X)$ may be estimated.

We focus here on inference on CATE-CVaR.
We provide analogous procedures for the lower bounds of \cref{thm:cvarlb1,thm:cvarlb2,thm:cvarlb3,cor:cvarlb4} in \cref{apx:lbest}.
Fix $\alpha$. Our inferential target is
$$\Psi=\cvarat\alpha(\cate)=\beta^*+\frac1{\alpha}\E(\cate-\beta^*)_-,\quad\text{where}~\beta^*=F_{\cate}^{-1}(\alpha)=\inf\fbraces{\beta:\Prb{\cate\leq\beta}\geq\alpha}.$$
Since $\cate$ is not directly observed, the first step is fitting it. Fortunately, recent advances in causal machine learning provide excellent tools for this \citep{slearner,xlearner,drlearner,rlearner,causaltree,causalforest}.
Given an estimate $\hat\catef$, we might consider a plug-in approach: $\hat\Psi^\text{plug-in}=\sup_\beta\fprns{\beta+\frac1{n\alpha}\sum_{i=1}^n(\hat\catef(X_i)-\beta)_-}$.
Unfortunately, the statistical behavior of $\hat\Psi^\text{plug-in}$ depends heavily on that of $\hat\catef$: if $\hat\catef$ converges slowly and/or has non-negligible bias, as occurs when fit by flexible machine-learning methods, both estimation rates and valid inference may be imperiled for $\hat\Psi^\text{plug-in}$.

\begin{algorithm}[t!]
\caption{Point estimate and confidence interval for $\cvarat\alpha(\cate)$}\label{alg:est}
\begin{algorithmic}[1]
\STATEx\textbf{Input:} Level $\alpha\in(0,1)$, data $\{(X_i,A_i,\Yobs_i):i=1,\dots,n\}$, number of folds $K$, $e,\mu,\tau$-estimators
\FOR{$k=1,\dots,K$}
\STATE{Estimate $\hat e^{(k)},\hat\mu^{(k)},\hat\tau^{(k)}$ using data $\{(X_i,A_i,\Yobs_i):i\not\equiv k-1~\text{(mod $K$)}\}$}
\STATE{Set $\hat\beta^{(k)}=\inf\fbraces{\beta:\sum_{i\not\equiv k-1~\text{(mod $K$)}}\fprns{\findic{\hat\tau^{(k)}(X_i)\leq \beta}-\alpha}\geq0}$\label{alg:est betastep}}
\STATE{\textbf{for}~$i\equiv k-1~\text{(mod $K$)}$~\textbf{do}~set $\phi_{i}=\phi(X_i,A_i,Y_i;\hat e^{(k)},\hat \mu^{(k)},\hat \tau^{(k)},\hat \beta^{(k)})
$\label{alg:est phistep}}
\ENDFOR
\STATE{Set $\hat\Psi=\frac1n\sum_{i=1}^n\phi_{i}$, $\hat{\op{se}}=\sqrt{\frac1{n(n-1)}\sum_{i=1}^n(\phi_{i}-\hat\Psi)^2}$\label{alg:est psistep}}
\STATE{Return $\hat\Psi$ as point estimate and $[\hat\Psi\pm \Phi^{-1}((1+\gamma)/2)\hat{\op{se}}]$ as $\gamma$-confidence intervals}
\end{algorithmic}
\end{algorithm}


Instead, we develop a debiasing approach that is \emph{insensitive} to CATE-estimation, accommodating both misspecified parametric models and flexible-but-imprecise machine-learning CATE-estimators. The main challenge is estimating $\beta^*$, which cannot be expressed by an estimating equation in $X,Y(0),Y(1)$, so its efficient/orthogonal estimation is unclear, unlike the case of quantile/CVaR treatment effects \citep{kallus2019localized,belloni2017program,firpo2007efficient}.
Fortunately, we care only about $\Psi$, not $\beta^*$, and special optimization structure in $\Psi$ gives robustness to perturbations. so even rough 
estimates suffice.
Our approach is therefore unique: we treat both $\catef$ and $\beta^*$ as nuisance parameters, together with $e,\mu$, and ensure simultaneous orthogonality to all four nuisances.

\Cref{alg:est} summarizes our procedure. It proceeds by approximating the sample average of $\Psi=\E\phi(X,A,Y,e,\mu,\catef,\beta^*)$, where, we define
\begin{equation}\label{eq:phi}
\phi(X,A,\Yobs;\tprop,\tmu,\tcate,\tbeta)=\tbeta+\frac1\alpha\indic{\tcate(X)\leq\tbeta}\prns{\tmu(X,1)-\tmu(X,0)+\frac{A-\tprop(X)}{\tprop(X)(1-\tprop(X))}\prns{\Yobs-\tmu(X,A)}-\tbeta}.
\end{equation}
We first estimate the unknown $(e,\mu,\tau,\beta^*)$. We do so using ``cross-fitting'' over $K$ even folds so that nuisance estimates are independent of samples where applied
\citep{schick1986,doubleML,zheng2011cross}.\footnote{We may avoid cross-fitting and fit nuisances once on the whole sample if we assume estimates belong to a Donsker class with probability tending to 1; we omit this option for brevity.}
As we discuss in detail in \cref{sec:cateest}, we treat $\catef$ as a separate nuisance even though $\catef(x)=\mu(x,1)-\mu(x,0)$. For one, this enables the use of specialized CATE-learners. We also treat $\beta^*$ as a separate nuisance (not as a parameter as in \citealp{kallus2019localized}) and fit it as the quantile of $\hat\catef(X)$ in the out-of-fold data. As simple regressions, $e$ and $\mu$ 
can be fit by parametric regression or standard machine-learning methods such as random forests, gradient boosting, neural networks, \etc.

\begin{remark}[Comparing different levels and inter-quantile averages]\label{remark:comparing}
To assess disparities, we may want to compare $\cvarat{\alpha}(\cate)$ to ATE (equivalently, $\cvarat1(\cate)$).
To get good confidence intervals on $\cvarat{\alpha'}(\cate)-\cvarat{\alpha}(\cate)$, we can replace $\phi_i$ in \cref{alg:est phistep} of \cref{alg:est} with the difference of $\phi_i$'s for $\alpha'$ and $\alpha$ (using the same nuisances except $\hat\beta^{(k)}$). Setting $\alpha'=1$, this will, in particular, correctly yield smaller confidence intervals on $\ate-\cvarat{\alpha}(\cate)$ for $\alpha$ near $1$. 
Similarly, if we want confidence intervals on inter-quantile average effects as in \cref{remark:gate}, then per \cref{eq:interquantile} we may simply replace $\phi_i$ in \cref{alg:est phistep} of \cref{alg:est} with the difference of $\phi_i$'s for $\alpha'$ and $\alpha$, weighted by $\frac{\alpha'}{\alpha'-\alpha}$ and $\frac{\alpha}{\alpha'-\alpha}$, respectively.
We may also consider covariances of $\phi_i$'s corresponding to many $\alpha$-levels for constructing simultaneous intervals.
\end{remark}

\begin{remark}[Partial-identification intervals]\label{remark:pi}
While \cref{alg:est} focuses on CATE-CVaR, which upper bounds ITE-CVaR, in \cref{apx:lbest} we provide inference procedures for lower bounds on ITE-CVaR. These can be combined to construct intervals containing ITE-CVaR with probability $\gamma$. By union bound, we can simply combine the one-sided $(1+\gamma)/2$-confidence intervals for the lower and upper bounds.
But coverage may be conservative ($>\gamma$) for the partial-identification interval given by the bounds. For calibrated $\gamma$-coverage (asymptotically), we must account for correlation between lower- and upper-bound estimates, given by the correlation between $\phi_i$'s for each procedure. Then, we can construct calibrated intervals following Appendix A.4 of \citet{kallus2021assessing}.
\end{remark}

\begin{remark}[Monotonicity]\label{remark:rearrangement}
While $\cvarat\alpha(\cate)$ is monotone in $\alpha$, \cref{alg:est}'s output for different $\alpha$ may not be due to estimation errors. We can post-process to ensure monotonicity using rearrangement \citep{hardy1952inequalities}, which only improves estimation and does not affect inference \citep{chernozhukov2010quantile}. We use this in \cref{sec:casestudy}.
\end{remark}

\subsection{Local Robustness and Confidence Intervals}\label{sec:localrobust}

We now establish favorable guarantees for \cref{alg:est}.
First, we show it is insensitive to slow but consistent estimation of nuisances, having first-order behavior as if we used true values. 
%

We will need some minimal regularity.

\begin{assumption}[Regularity]\label{asm:regularity}
$\bar e\leq e\leq1-\bar e$ and
$\abs{Y}\leq B$ for positive constants $\bar e,\,B>0$.\break
$F_{\cate}$ is continuously differentiable at $F^{-1}_{\cate}(\alpha)$.
\end{assumption}

The first condition ensures that the $X$-distributions of experimental groups \emph{overlap}. It is usually guaranteed in randomized experiments by setting $e(X)$ constant ($A\indep X$). In unconfounded observational studies, it is a standard assumption. 
The second condition requires bounded outcomes and is largely technical to make analysis tractable.
The third condition prohibits degeneracy of the quantile. The same is needed for asymptotic normality of sample quantiles of \emph{observed} variables. If 
$\cate$ is discrete, the condition may be replaced by $\exists \varepsilon>0:F^{-1}_{\cate}(\alpha-\varepsilon)=F^{-1}_{\cate}(\alpha+\varepsilon)$, yielding superefficient quantile estimation. The only problematic case is multiplicity of $\{\beta:F_{\cate}(\beta)=\alpha\}$,
but only finitely-many such ``bad'' $\alpha$'s exist.
Since the focus is on $X$ being rich, we focus on the continuous case and the condition in \cref{asm:regularity}.


We first show how, under \cref{asm:regularity}, estimation rates for $\hat\tau^{(k)}$ translate to rates for $\hat\beta^{(k)}$. 

\begin{lemma}\label{lemma:betalemma}
Suppose \cref{asm:regularity} holds.
Then, for each $k=1,\dots,K$, $\hat\beta^{(k)}$ in \cref{alg:est betastep} of \cref{alg:est} satisfies
$$
\fabs{\hat\beta^{(k)}-\beta^*}=O_p(n^{-1/2}\vee \fmagd{\hat\tau^{(k)}-\tau}_r^{\frac{r}{r+1}})\quad\forall r\in[1,\infty].
$$
\end{lemma}





We now show robust oracle-like behavior for $\hat\Psi$.

\begin{theorem}\label{thm:asympnormal}
Suppose \cref{asm:regularity} holds and that for $k=1,\dots,K$,
$\fmagd{\hat e^{(k)}-e}_2=o_p(1)$,
$\fmagd{\hat\mu^{(k)}-\mu}_2=o_p(1)$,
$\fmagd{\hat e^{(k)}-e}_2\fmagd{\hat\mu^{(k)}-\mu}_2=o_p(n^{-\frac{1}{2}})$,
$\fmagd{\hat\tau^{(k)}-\tau}_\infty=o_p(n^{-\frac{1}{4}})$,
$\fPrb{\fmagd{\hat\mu^{(k)}}_\infty\leq B}\to1$,
and
$\fPrb{\bar e\leq 
\hat e^{(k)}
\leq1-\bar e}\to1$.
Then $\hat\Psi,\,\hat{\op{se}}$ in \cref{alg:est psistep} of \cref{alg:est} satisfy
\begin{align*}
&\hat\Psi=\frac1n\sum_{i=1}^n\phi(X,A,Y;e,\mu,\tau,\beta^*)+o_p(n^{-1/2})=\Psi+O_p(n^{-1/2}),\\
&\fPrb{\Psi\in[\hat\Psi\pm \Phi^{-1}((1+\gamma)/2)\hat{\op{se}}]}\to\gamma~~\forall \gamma.
\end{align*}
\end{theorem}

The rate assumptions on $e$ and $\mu$ are lax:
it suffices to have $o_p(n^{-1/4})$-rates on both or 
no rate on $\mu$ at all if $e$ is known.
This parallels standard conditions in double-machine-learning ATE-estimation, achievable by a variety of machine-learning methods \citep{doubleML}.
We explore the condition on $\catef$ in \cref{sec:cateest}.

\subsection{Double Robustness and Double Validity}\label{sec:caterobust}

\Cref{thm:asympnormal} guarantees good performance if all nuisances are estimated slowly, but still consistently. 
But even if nuisances are inconsistent, we perform well.
%
%
%

First, we establish a 
property mirroring doubly-robust ATE-estimation \citep{RRZDoubleRobust}:
even if $e$ or $\mu$ are inconsistent, we remain consistent, provided $\catef$ is consistently estimated, albeit slowly.

\begin{theorem}[Double robustness]\label{thm:doublyrobust}
Fix any $\sprop,\smu$ with
$\bar e\leq \sprop\leq 1-\bar e$, $\fmagd{\smu}_\infty\leq B$.
Let $r_n\to0$ be a deterministic sequence.
Suppose \cref{asm:regularity} holds and that for $k=1,\dots,K$,
$\fmagd{\hat e^{(k)}-\sprop}_2=o_p(1)$,
$\fmagd{\hat\mu^{(k)}-\smu}_2=o_p(1)$,
$\fmagd{\hat\tau^{(k)}-\tau}_\infty=O_p(r^{1/2}_n)$,
$\fPrb{\fmagd{\hat\mu^{(k)}}_\infty\leq B}\to1$,
$\fPrb{\bar e\leq 
\hat e^{(k)}
\leq1-\bar e}\to1$,
and
$$\text{either}\quad\fmagd{\hat e^{(k)}-e}_2=O_p(r_n)\quad\text{or}\quad\fmagd{\hat\mu^{(k)}-\mu}_2=O_p(r_n).$$
Then $\hat\Psi$ in \cref{alg:est psistep} of \cref{alg:est} satisfies:
$$
\hat\Psi=\Psi+O_p(r_n\vee n^{-1/2}).
$$
\end{theorem}

\Cref{thm:doublyrobust} is particularly strong in experiments ($e$ known): we can get away with $\hat\mu^{(k)}=0$. We need only estimate CATE at $o_p(n^{-1/4})$-rates to ensure $O_p(n^{-1/2})$-consistency.

It would appear we must consistently estimate CATE to have hope of estimating its CVaR. 
While true, we next show that \emph{even} if we mis-estimate CATE \emph{and also} one of $e,\mu$, we \emph{still} get an upper bound on CATE-CVaR (hence on ITE-CVaR). This appears to be the second finding of a double-validity property since being first documented
in sensitivity analysis \citep{dorn2021doubly}.

We first establish the population-level bound behavior and then state the implication for estimation.
\begin{lemma}\label{lemma:popdoublevalid}
Fix any $\scate:\X\to\Rl$. Let $\sbeta=F_{\scate(X)}^{-1}(\alpha)$.
Suppose \cref{asm:regularity} holds with $\catef$ replaced with $\scate$. Then:
\begin{equation}\label{eq:popvalid}
\cvarat\alpha(\cate)\leq
\sbeta+\frac1\alpha\Efb{\findic{\scate(X)\leq\sbeta}(\cate-\sbeta)}.
\end{equation}
\end{lemma}

\begin{theorem}[Double validity]\label{thm:doublyvalid}
Fix any $\sprop,\smu,\scate$ with
$\bar e\leq \sprop\leq 1-\bar e$, $\fmagd{\smu}_\infty\leq B$, $\fmagd{\scate}_\infty\leq 2B$.
Let $r_n\to0$ be a deterministic sequence.
Suppose \cref{asm:regularity} holds with $\catef$ replaced with $\scate$ and that for $k=1,\dots,K$,
$\fmagd{\hat e^{(k)}-\sprop}_2=o_p(1)$,
$\fmagd{\hat\mu^{(k)}-\smu}_2=o_p(1)$,
$\fmagd{\hat\tau^{(k)}-\scate}_\infty=O_p(r_n)$,
$\fPrb{\fmagd{\hat\mu^{(k)}}_\infty\leq B}\to1$,
$\fPrb{\bar e\leq 
\hat e^{(k)}
\leq1-\bar e}\to1$,
and
$$\text{either}\quad\fmagd{\hat e^{(k)}-e}_2=O_p(r_n)\quad\text{or}\quad\fmagd{\hat\mu^{(k)}-\mu}_2=O_p(r_n).$$
Then $\hat\Psi$ in \cref{alg:est psistep} of \cref{alg:est} satisfies:
$$
\hat\Psi\geq\Psi-O_p(r_n\vee n^{-1/2}).
$$
\end{theorem}

\Cref{thm:doublyvalid} guarantees extensive robustness and suggests a practical, blackbox-free approach in experimental settings: set $\hat\mu^{(k)}=0$ and use simple \emph{misspecified} parametric models (\eg, linear) for CATE-estimation, and we still estimate a valid ITE-CVaR bound at fast $O_p(n^{-1/2})$-rates.

\subsection{CATE-Estimation and Rates}\label{sec:cateest}

\Cref{alg:est} accepts separate learners for \emph{both} $\mu$ \emph{and} $\catef$. So, while $\catef(x)=\mu(x,1)-\mu(x,0)$, we need \emph{not} have $\hat\catef^{(k)}(X)=\hat\mu^{(k)}(x,1)-\hat\mu^{(k)}(x,0)$, and in fact we should not. Recent work advocates and provides specialized methods for \emph{directly} estimating CATE \citep{causaltree,causalforest,xlearner,slearner,rlearner,drlearner}.

This is important because \cref{alg:est} uses the $\mu$- and $\tau$-estimates differently and, correspondingly, our theoretical results impose different assumptions on each. 
The $\tau$-estimate is used for approximating the event $\indic{\tau(X)\leq\beta^*}$, which is crucial for targeting CVaR correctly.
In contrast, the $\mu$-estimate is just used in order to estimate a weighted-average treatment effect, given the weights $\indic{\tau(X)\leq\beta^*}$, and is therefore interchangeable with propensity.

We next review different options for CATE-estimation and how these ensure the conditions of \cref{thm:asympnormal,thm:doublyrobust,thm:doublyvalid}. We emphasize that these need not be understood as exhaustive list of which learners to use: 
practically, the nuisance-estimation rates are high-level assumptions that generally say one may safely plug-in black-box machine-learning estimators to \cref{alg:est}: no restrictions are made but rates (no metric-entropy conditions), estimators can be flexible/nonparametric in that rates can be much slower than ``parametric" $O_p(n^{-1/2})$-rates, and results are exceedingly robust to inconsistent estimation.



\subsubsection{Experimental settings}\label{sec:pseudoexp}

A major issue with CATE-estimation by differencing outcome regressions is that effect signals are easily lost. CATE is generally simpler and less variable than baseline mean outcomes, $\mu(X,0),\mu(X,1)$. For example, many variables often help predict outcomes, but few modulate the treatment effect. It is therefore imperative to learn CATE directly. 

In experimental settings ($e$ known) we can construct a pseudo-outcome $\Delta=\frac{A-e(X)}{e(X)(1-e(X))}\Yobs$ and, since $\cate=\Eb{\Delta\mid X}$, learn CATE by regressing $\Delta$ on $X$, using any supervised-learning method.
One case that theoretically ensures $\fmagd{\hat \tau^{(k)}-\tau}_\infty=o_p(n^{-1/4})$ is when $\catef(x)$ is more-than-$d/2$-smooth in $x\in\R d$ \citep[Theorem 1]{stoneglobal}.
Another option is $\catef(x)$ linear with $o(\sqrt{n}/\log d)$-nonzero coefficients \citep{belloni2017program}.
Note this works \emph{regardless} of $\mu$ being nice.

Or, we may avoid black-box models (and cross-fitting) altogether by using simple linear regression of $\Delta$ on $X$ to obtain a valid bound per \cref{thm:doublyvalid}.


To satisfy the other conditions,
for \cref{thm:doublyrobust,thm:doublyvalid} we can set $\mu=0$,
and for \cref{thm:asympnormal} we need only estimate $\mu$ consistently without rate. We can either estimate $\mu$ directly or only estimate $\bar\mu(X)=\Eb{Y\mid X}$ and set $\hat\mu^{(k)}(X,A)=\hat{\bar\mu}^{(k)}(X)+(A-e(X))\hat\tau^{(k)}(X)$.
Consistency for either is immediate from
$\E\Yobs^2<\infty$ \citep
{gyorfi2002distribution}.


\subsubsection{Observational settings}

When $e$ is unknown, the pseudo-outcome-construction needs refinement.
One option is DR-leaner \citep{drlearner}:
regress $\Delta=\hat\mu(X,1)-\hat\mu(X,0)+\frac{A-\hat e(X)}{\hat e(X)(1-\hat e(X))}(\Yobs-\hat\mu(X,A))$ on $X$, where $\hat e,\hat \mu$ are appropriately cross-fitted. 
Another is R-learner \citep{rlearner}:
let $\hat\catef$ minimize the average of $\prns{\Yobs-\hat{\bar\mu}(X)-(A-\hat e(X))\hat\catef(X)}^2$, where $\hat e,\hat{\bar\mu}$ are appropriately cross-fitted.
\citet[Corollary 3]{drlearner} provides rates for local-polynomial R-learners: if $e(x)$ is $s_e$-smooth in $x\in\R d$, $\bar\mu(x)$ $s_\mu$-smooth, and $\tau(x)$ more-than-$d/2$-smooth, then we obtain $o_p(n^{-1/4})$-rate pointwise error, provided $s_e\geq s_\mu,\,\frac{s_e+s_\mu}{2}>\frac{d}{8}$. To convert pointwise-error bounds to sup-norm-error bounds, we may follow the discretization approach of \citet{stoneglobal}, incurring only logarithms.
Or, we can simply use linear R- or DR-learners and get a valid bound per \cref{thm:doublyvalid}.


\section{Case Study}\label{sec:casestudy}

\begin{figure}[t!]\centering%
\begin{minipage}{0.333\textwidth}
\centering
\includegraphics[width=\linewidth]{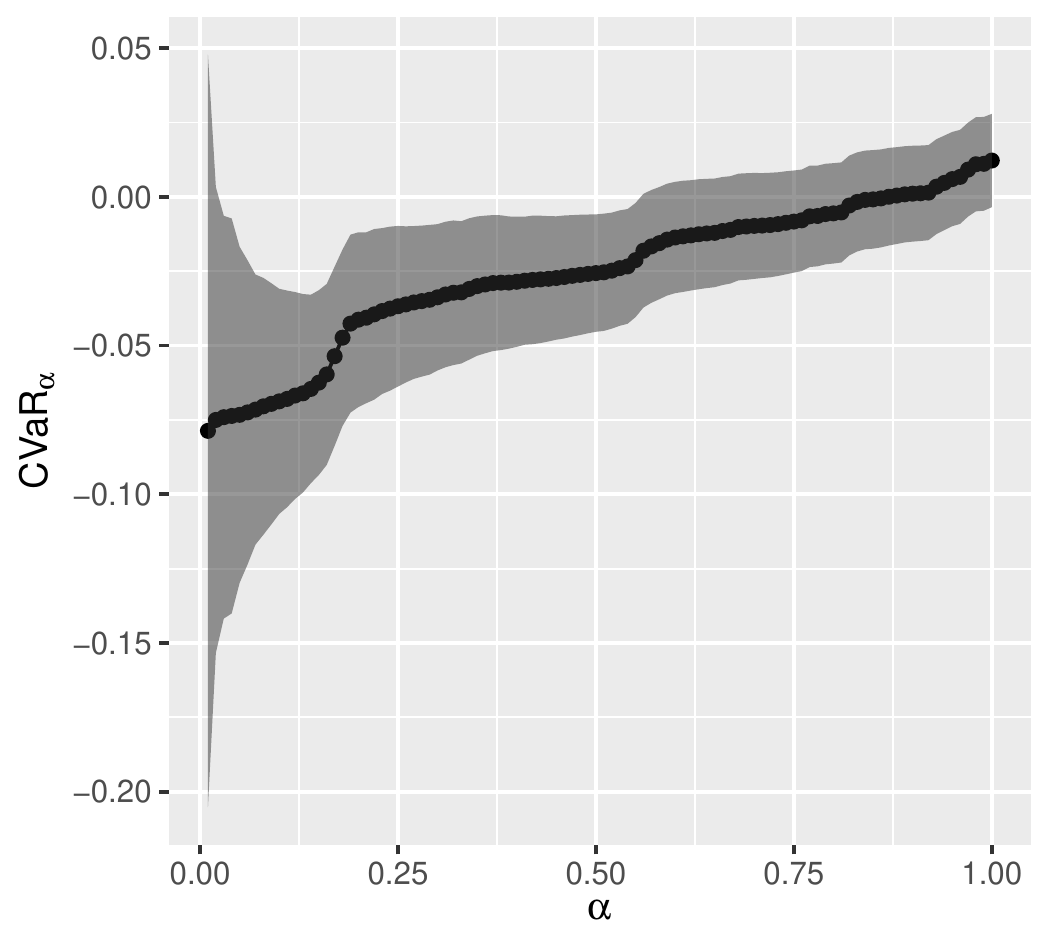}
\vspace{-1.5em}\caption{$\cvarat\alpha(\cate)$}\label{fig:job_cvar}
\end{minipage}\hfill%
\begin{minipage}{0.333\textwidth}
\centering
\includegraphics[width=\linewidth]{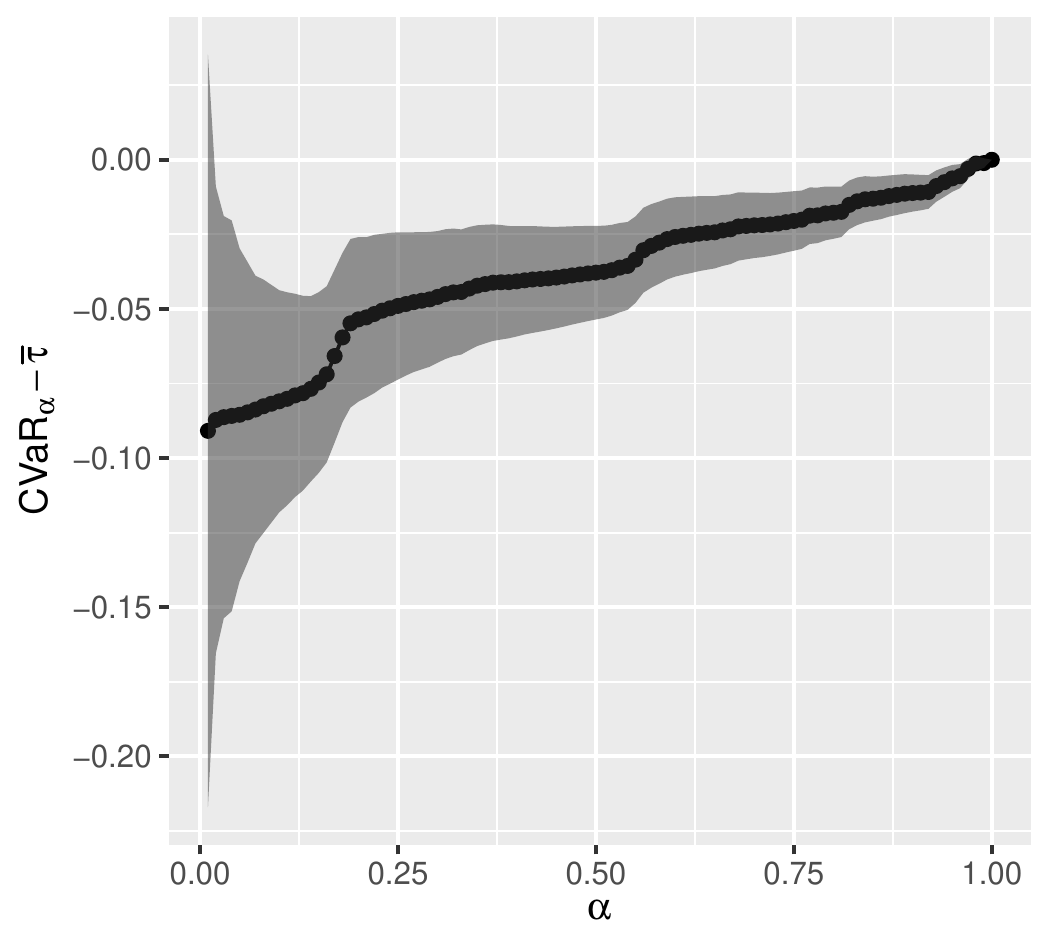}
\vspace{-1.5em}\caption{$\cvarat\alpha(\cate)-\ate$}\label{fig:job_cvarvsate}
\end{minipage}\hfill%
\begin{minipage}{0.333\textwidth}
\centering
\includegraphics[width=\linewidth]{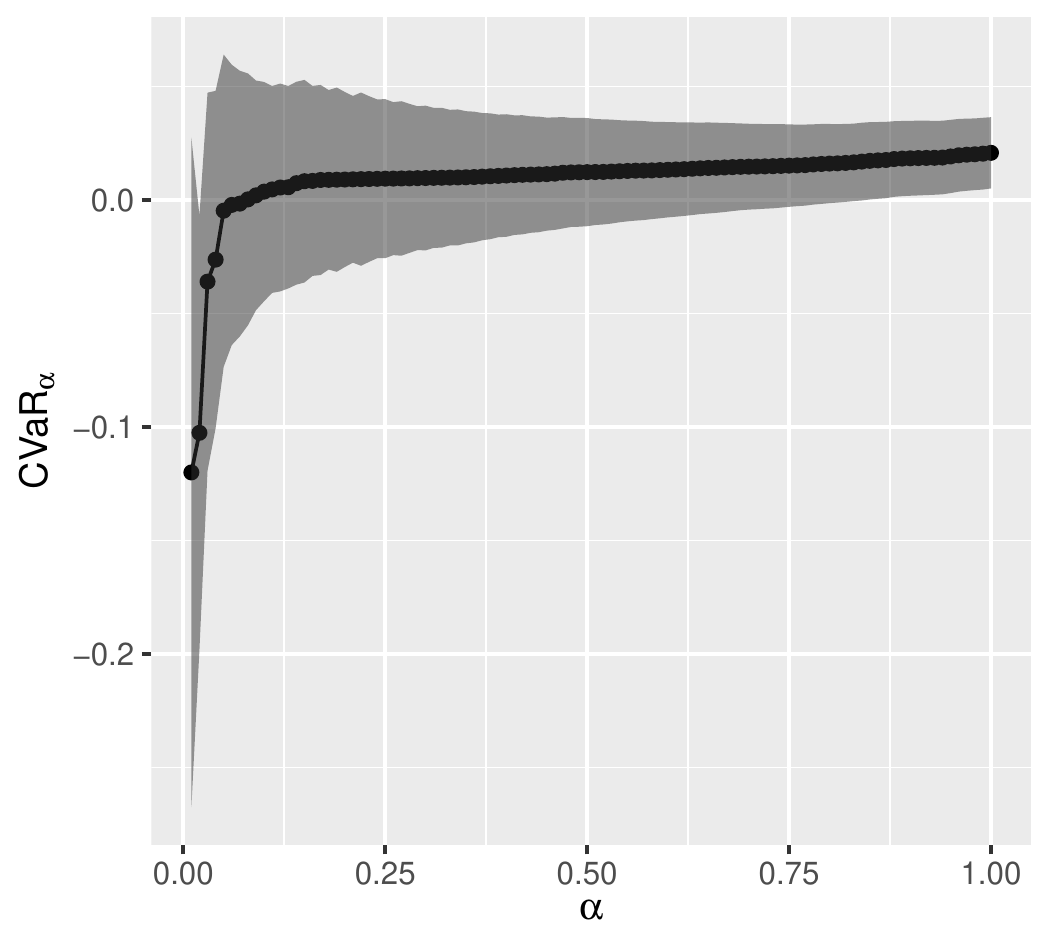}
\vspace{-1.5em}\caption{$\cvarat\alpha(\catef_1(X_1))$}\label{fig:job_cvar_bad}
\end{minipage}\\[0.75\baselineskip]%
\begin{minipage}{0.5\textwidth}
\centering
\includegraphics[width=\linewidth]{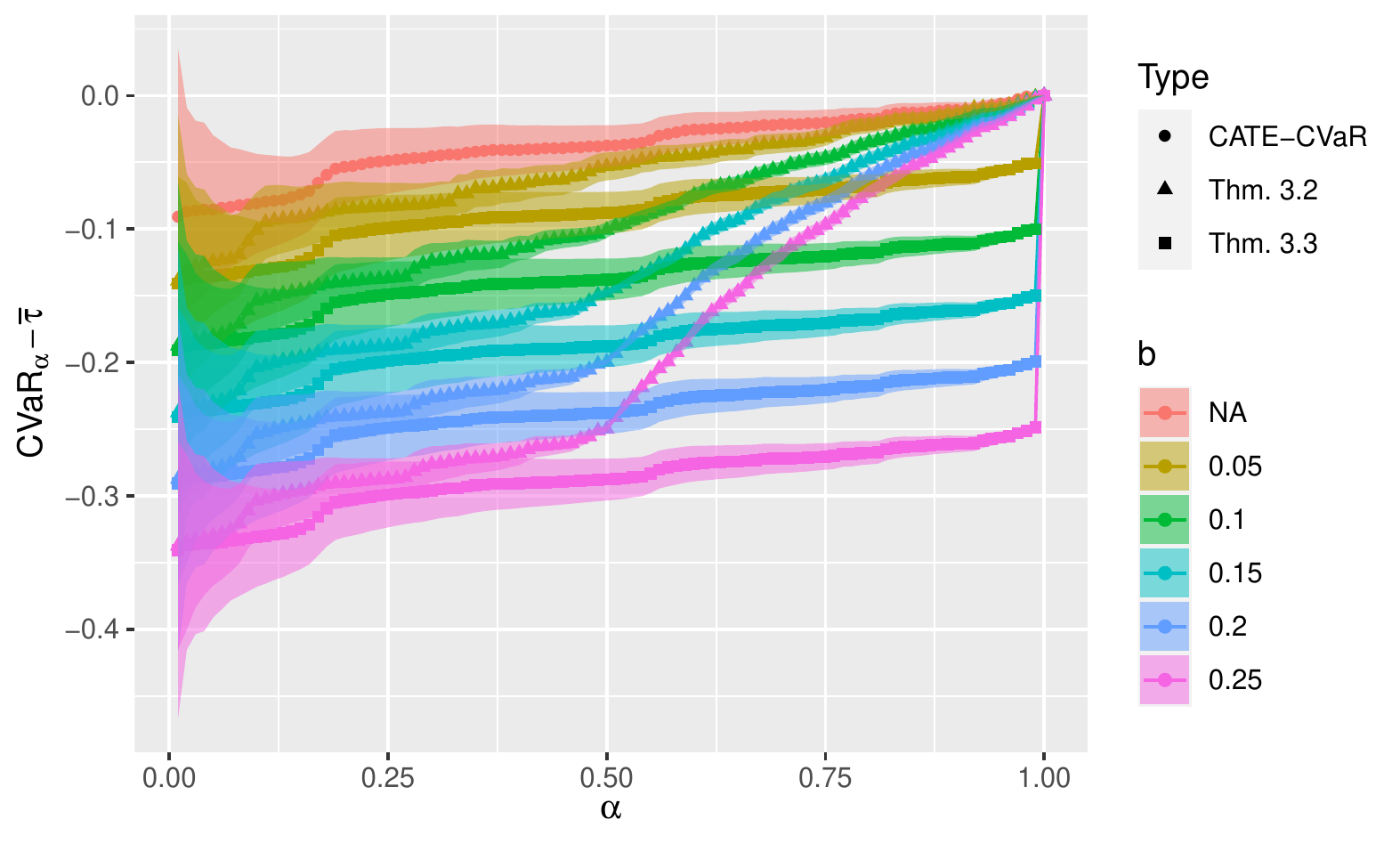}
\vspace{-1.5em}\caption{Bounds based on\break residual-heterogeneity range}\label{fig:job_bounded_bounds}
\end{minipage}\hfill%
\begin{minipage}{0.5\textwidth}
\centering
\includegraphics[width=\linewidth]{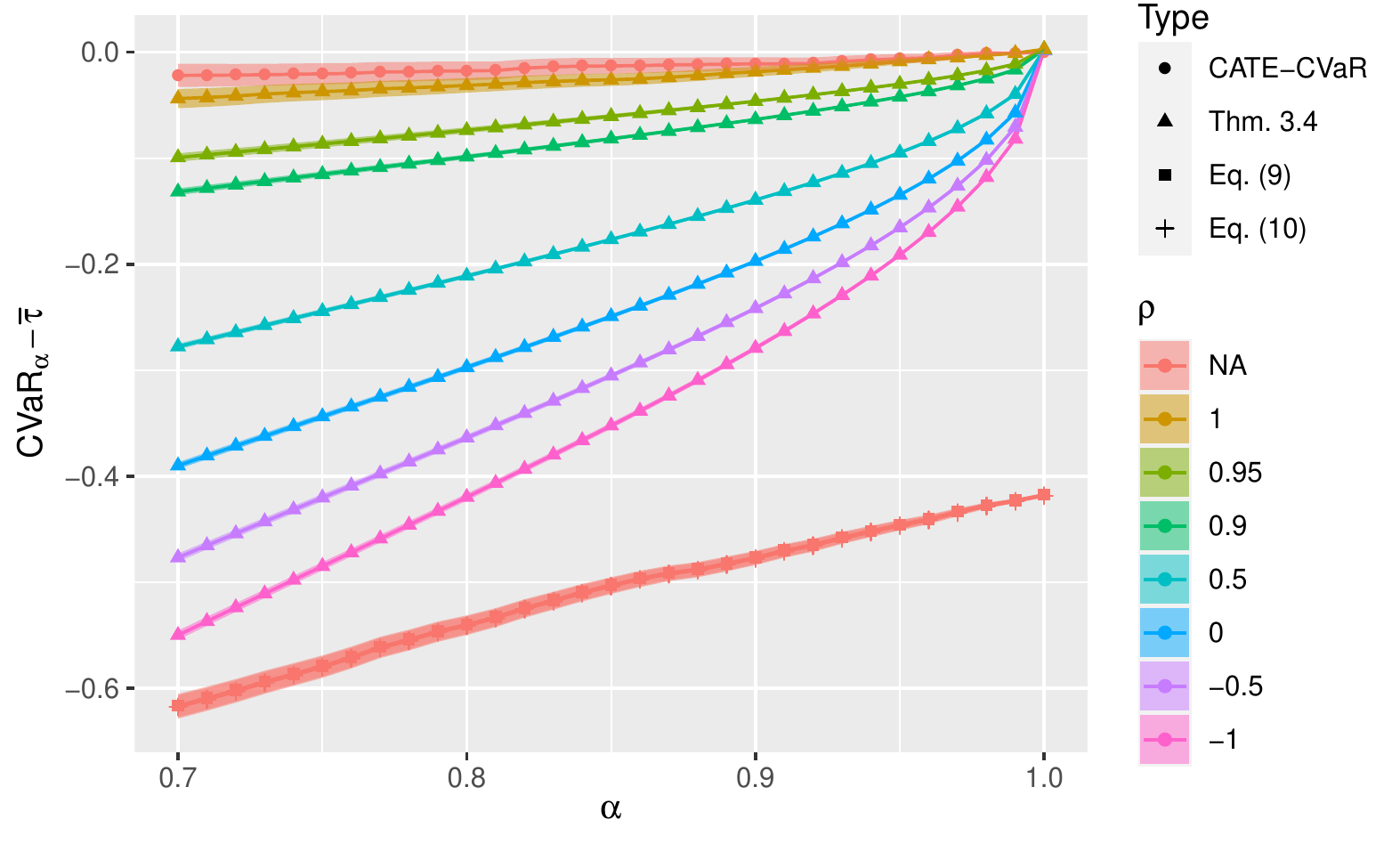}
\vspace{-1.5em}\caption{Bounds based on\break residual-heterogeneity variance}\label{fig:job_condvar_bounds}
\end{minipage}
\end{figure}

We now demonstrate our bounds and inference.\footnote{Replication code is available at \url{https://github.com/CausalML/TreatmentEffectRisk}.} 
While we consider a program-evaluation example,
we believe our results are also particularly relevant to A/B testing on online platforms, where, after testing, product innovations are usually either scrapped/reworked or broadly rolled out, and where ATEs are often small, creating an opportunity for many users to be negatively impacted despite positive average effects. Little such data is public, however.

\subsection{Background and Setup}

\citet{behaghel2014private} analyze a large-scale randomized experiment comparing assistance programs offered to French unemployed individuals.
They compare three arms: individuals in the ``control'' arm receive the standard services of the Public Employment Services, in ``public'' receive an intensive counseling program run by a public agency, and in ``private'' a similar program run by private agencies.

We consider a hypothetical scenario where the private-run counseling program  ($A=0$) is currently being offered to the unemployed and we consider the change to a public-run program  ($A=1$).\footnote{Some individuals assigned to the additional counseling refused it. We nonetheless restrict our attention to intent-to-treat interventions, considering hypothetically making available either the public-run or private-run counseling to unemployed individuals, who may decline it.}
We take reemployment within six months as our (binary) outcome.

The ATE is $1.22$ percentage points (90\%-CI $[-0.35, 2.8]$), a $4.9\%$ increase in reemployment.
This suggests a positive/neutral effect, so a policymaker might hypothetically consider this an acceptable policy change, \eg, if the public-run program provided cost savings.\footnote{\citet[section IV]{behaghel2014private} discuss why public-run programs fare better.}
%

To apply our methodology, we consider all pre-treatment covariates in table~2 of \citet{behaghel2014private}, except we treat as numeric (rather than dichotomize) age, number children, years experience, salary target, assignment timing, and number unemployment spells. Other variables quantify education, employment level and type, gender, martial status, national origin, region, unemployment reason, and long-term-unemployment risk. 
The propensity is constant.
As recommended in \cref{sec:pseudoexp}, we fit CATE using a pseudo-outcome linear regression. We estimate $\mu$
using cross-fitted gradient-boosting machines.

\subsection{Upper bounds}

\Cref{fig:job_cvar} presents inference on CATE-CVaR using \cref{alg:est} for $\alpha\in\{0.01,0.02,\dots,1\}$. The line represents our point estimate, after rearrangement as recommended in \cref{remark:rearrangement},\footnote{We present the figure without rearrangement in \cref{apx:casestudy}.} and the shaded region represents point-wise 90\%-confidence intervals. 
Note uncertainty grows for smaller $\alpha$.

We see that the ATE-estimate (right-most point) is positive with an interval containing zero. 
We find, however, that some 56\%-sized $X$-defined-subpopulation has a negative effect at 90\%-confidence.\footnote{Since outcome is binary, the \emph{largest} fraction that can have a negative effect is $(50\times(1-\ate))\%$, so either $\ate<0$ or at most half may be negatively affected. The ATE interval indeed contains zero with confidence only 90\%.}
This strongly suggests that the change, if enacted could materially negatively impact a large portion of the population, despite the positive/neutral ATE.
Thus, considering treatment effect \emph{risk} provides a crucial metric not reflected in the ATE.
This risk is also \emph{not} reflected in DTEs: the binary potential-outcome distributions are \emph{fully} specified by just $\E[Y(0)],\,\E[Y(1)]$.\footnote{In particular, the $\alpha$-quantile DTE is uselessly \emph{zero} for \emph{all} $\alpha\in[0,1]\backslash\{1-\E[Y(0)],1-\E[Y(1)]\}$ and the $\alpha$-CVaR DTE is $\frac1\alpha(\E[Y(1)]-1+\alpha)_+-\frac1\alpha(\E[Y(0)]-1+\alpha)_+$, which is not even monotonic. For illustration we plot it in \cref{apx:casestudy}.}




In \cref{fig:job_cvarvsate} we focus on comparing CATE-CVaR to ATE following \cref{remark:comparing}. The only difference to \cref{fig:job_cvar} is a slight vertical shift and that confidence intervals (correctly) shrink to a point as $\alpha\to1$, enabling more confident conclusions comparing subpopulations to the population.


In \cref{fig:job_cvar_bad} we consider CATE-CVaR when we capture less heterogeneity,
using only age, high-school dropout, African national origin, and Paris-region resident as covariates ($X_1$). This detects no significant risk.




\subsection{Lower bounds}

While the upper bounds show a significant subpopulation can be negatively harmed, being only bounds, it may be the subpopulation can be harmed even more or an even larger subpopulation can be harmed. Lower bounds help us understand how much greater the risk might be.

In \cref{fig:job_bounded_bounds} we consider our lower bounds (vs ATE) when limiting the residual-heterogeneity range given by Theorems \ref{thm:cvarlb2} (two-sided range) and \ref{thm:cvarlb1} (one-sided range).

Since it may be hard to justify and calibrate a limited range, in \cref{fig:job_condvar_bounds} we consider lower bounds given by \cref{thm:cvarlb3,cor:cvarlb4} by limiting residual-heterogeneity variance. For the former, we fit $\op{Var}(Y\mid A,X)$ using gradient-boosting machines and construct $\bar\sigma^2(X)$ per \cref{eq:vardecomp} by varying constant values of $\rho(X)=\rho\in[-1,1]$. Recall $\rho=-1$ always yields an assumption-free bound. We use the same model to estimate the right-hand side of \cref{eq:cvarlb4b}. We compute the cross-validated root-mean-squared prediction error to estimate the right-hand side of \cref{eq:cvarlb4c}.

We observe that assuming perfectly-conditionally-correlated potential outcomes yields a lower bound very close to the upper bound.
The bounds of \cref{cor:cvarlb4} appear loose; indeed they are not tight.

\section{Concluding Remarks}

We study the average effect on those worst-affected by a proposed change as a measure of its \emph{risk}, how to tightly bound it given covariates that explain some heterogeneity, and how to make robust inferences on these bounds even when this heterogeneity is roughly estimated. This provides very practical tools for assessing policy and product changes beyond their ATE and DTEs. We can safely use flexible yet biased/slow-to-converge machine learning, or we can avoid black-box models and easily get good bounds by considering only linear projections of heterogeneity .
In the hypothetical case study this detected that, what appeared to be a positive/neutral change could actually very negatively impact a substantial subpopulation. 

We focused on experimental (or, unconfounded observational) settings without interference, where risk is already unidentifiable \emph{despite} randomization. A future direction is to consider the impact of interference \citep{johari2022experimental,athey2018exact}
or confounding \citet{tan2006}, where even ATEs are unidentifiable and fairness is harder to assess \citep{kilbertus2020sensitivity,jung2020bayesian,kallus2018residual}. Interestingly, for partial identification under \citet{tan2006}' model, $X$-conditional outcome-CVaR plays a crucial role \citep{dorn2021doubly}. Another direction may be to consider other risk measures, such as given by Kullback-Leibler ambiguity sets \citep{ahmadi2012entropic}. Per \cref{footnote:coherentrisk}, the tight upper bound is still the risk measure applied to CATE, but it remains to compute lower bounds and design robust inference methods.

\begin{acks}
I thank Netflix's Dar\'io Garc\'ia Garc\'ia, Molly Jackman, Danielle Rosenberg, William Nelson, and Martin Tingley for very helpful conversations.
\end{acks}

\bibliographystyle{ACM-Reference-Format}
\bibliography{references}

\newpage
\appendix

\makeatletter
\vbox{\@titlefont\centering Appendices}
\makeatother

\section{Inference for Lower Bounds}\label{apx:lbest}

In the main text we focused on inference for CATE-CVaR because of the primary importance of an upper bound and because it is a quantity of interest independent of being a bound by virtue of summarizing heterogeneous treatment effects. 
Here we extend our inference procedure (\cref{alg:est}) to the lower bounds of \cref{thm:cvarlb1,thm:cvarlb2,thm:cvarlb3,cor:cvarlb4}.

\subsection{Inference for the Lower Bounds of \cref{thm:cvarlb1,cor:cvarlb4}}

Consider $\Psi$ equal to the right-hand side of \cref{eq:cvarlb1}
or any of the lower bounds in \cref{cor:cvarlb4}
with $\alpha$ fixed.
This is the simplest extension: the estimand is simply the CATE-CVaR minus a constant. All we need to do is run \cref{alg:est} and subtract the constant.

\subsection{Inference for the Lower Bound of \cref{thm:cvarlb2}}

Consider $\Psi$ equal to the right-hand side of \cref{eq:cvarlb2} with $\alpha$ fixed.
This case also follows easily from our CATE-CVaR procedure since the bound is simply the CVaR of the equal-parts mixture of the distribution of CATE minus $b$ and the distribution of CATE plus $b$. Thus, we need only to make two changes to \cref{alg:est}. First, in \cref{alg:est betastep}, set $\hat\beta^{(k)}=\inf\fbraces{\beta:\sum_{i\not\equiv k-1~\text{(mod $K$)}}\fprns{\findic{\hat\tau^{(k)}(X_i)\leq \beta-b}+\findic{\hat\tau^{(k)}(X_i)\leq \beta+b}-2\alpha}\geq0}$. Second, in \cref{alg:est phistep}, we should set $\phi_i=\frac12(\phi_i^++\phi_i^-)$ where $\phi^\pm_i$ is as is currently in \cref{alg:est phistep} but with $\hat\beta^{(k)}$ changed to $\hat\beta^{(k)}\pm b$, respectively for $+$ and $-$.

\subsection{Inference for the Lower Bound of \cref{thm:cvarlb3}}

Consider $\Psi$ equal to the right-hand side of \cref{eq:cvarlb3} with $\alpha$ and $\bar\sigma^2(x)$ fixed.
This case also just requires a few changes to \cref{alg:est}. 
First, in \cref{alg:est betastep}, set $\hat\beta^{(k)}$ to the maximizer of $\hat f^{(k)}(\beta)=\beta+\frac1{2\alpha}\frac{1}{\abs{\braces{i\not\equiv k-1~\text{(mod $K$)}}}}\sum_{i\not\equiv k-1~\text{(mod $K$)}}\prns{\hat\catef^{(k)}(X_i)-\beta-\sqrt{(\hat\catef^{(k)}(X_i)-\beta)^2+\bar\sigma^2(X_i)}}$. This can be done using golden-section search.
Second, in \cref{alg:est phistep}, set 
$\phi_i=\hat\beta^{(k)}+\frac{1}{2\alpha}\prns{\hat\catef^{(k)}(X_i)-\hat\beta^{(k)}-\sqrt{(\hat\catef^{(k)}(X_i)-\hat\beta^{(k)})^2+\bar\sigma^2(X_i)}}+\frac1{2\alpha}\prns{1-\prns{\hat\catef^{(k)}(X_i)-\hat\beta^{(k)}}{\prns{(\hat\catef^{(k)}(X_i)-\hat\beta^{(k)})^2+\bar\sigma^2(X_i)}^{-1/2}}}\frac{A_i-e(X_i)}{e(X_i)(1-e(X_i))}(Y_i-\hat\mu^{(k)}(X_i,A_i))$.

If $\bar\sigma^2(x)$ is given by the hand-side of \cref{eq:vardecomp} where $\sigma^2(X,A)=\op{Var}(Y\mid X,A)$ is estimated rather than known and where $\rho(X)$ is fixed (\eg, constant), then we need to cross-fit $\hat\sigma^{2,(k)}(X,A)$ and use the resulting $\hat{\bar\sigma}^{2,(k)}(x)$ in $f^{(k)}$ above as well as add the following term to the $\phi_i$ in the previous paragraph in order to account for the additional uncertainty: $-\frac12\prns{1-\rho(X_i)\prns{\hat\sigma^{2,(k)}(X_i,0)/\hat\sigma^{2,(k)}(X_i,1)}^{2A_i-1}}\prns{(\hat\catef^{(k)}(X_i)-\hat\beta^{(k)})^2+\hat{\bar\sigma}^{2,(k)}(X_i)}^{-1/2}\prns{\frac{A_i}{\hat e^{(k)}(X_i)}+\frac{1-A_i}{1-\hat e^{(k)}(X_i)}}\break\prns{Y_i^2-\hat\sigma^{2,(k)}(X_i,A_i)-\mu^{(k)}(X_i,A_i)^2-2\mu^{(k)}(X_i,A_i)(Y_i-\mu^{(k)}(X_i,A_i))}$.

\section{Additional Results for Section~\ref{sec:casestudy}}\label{apx:casestudy}

\begin{figure}[t!]\centering%
\begin{minipage}{0.5\textwidth}
\centering
\includegraphics[width=.9\linewidth]{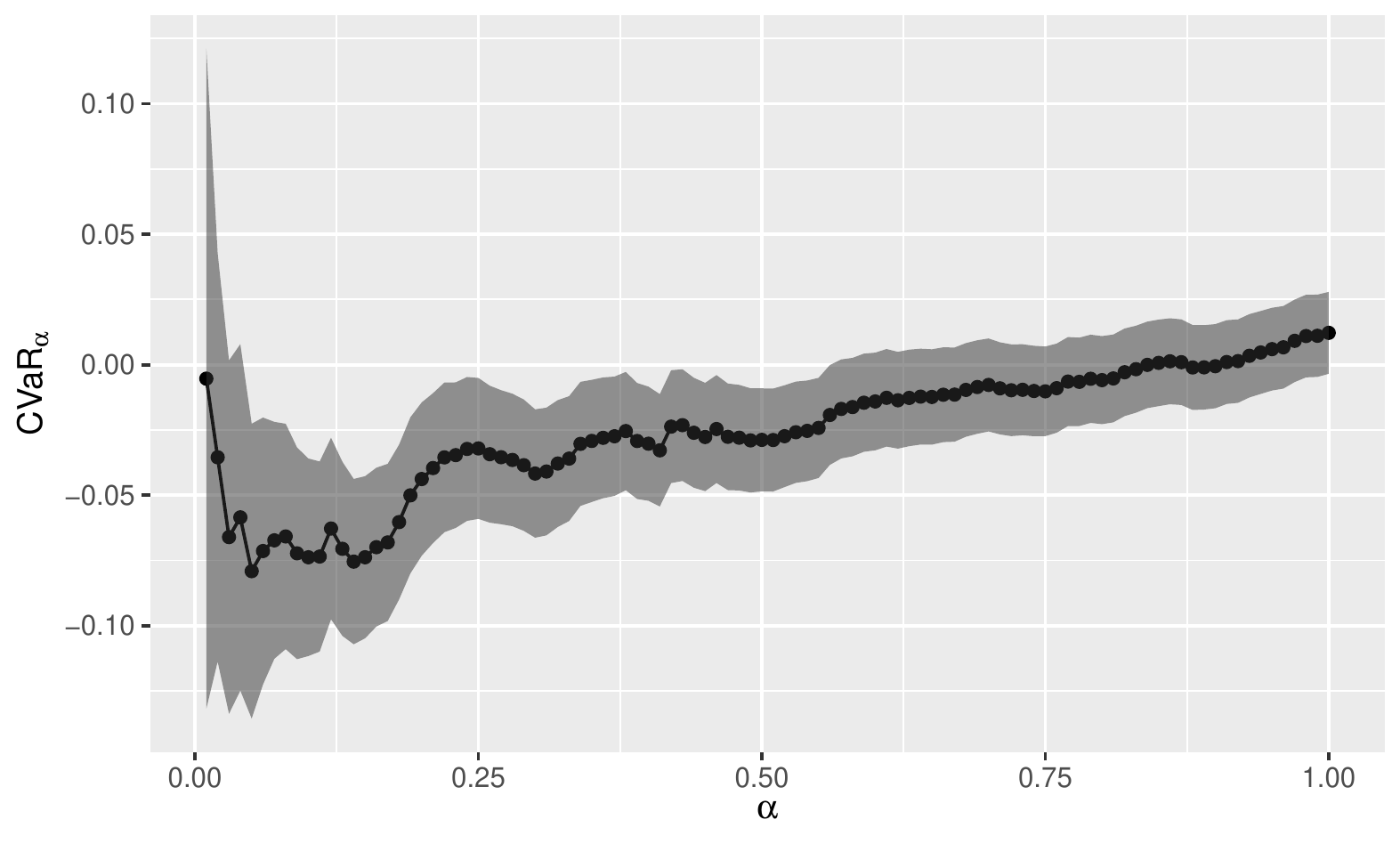}
\caption{The results from \cref{fig:job_cvar} \emph{without} rearrangement}\label{fig:job_cvar_notrearranged}
\end{minipage}\hfill%
\begin{minipage}{0.5\textwidth}
\centering
\includegraphics[width=.9\linewidth]{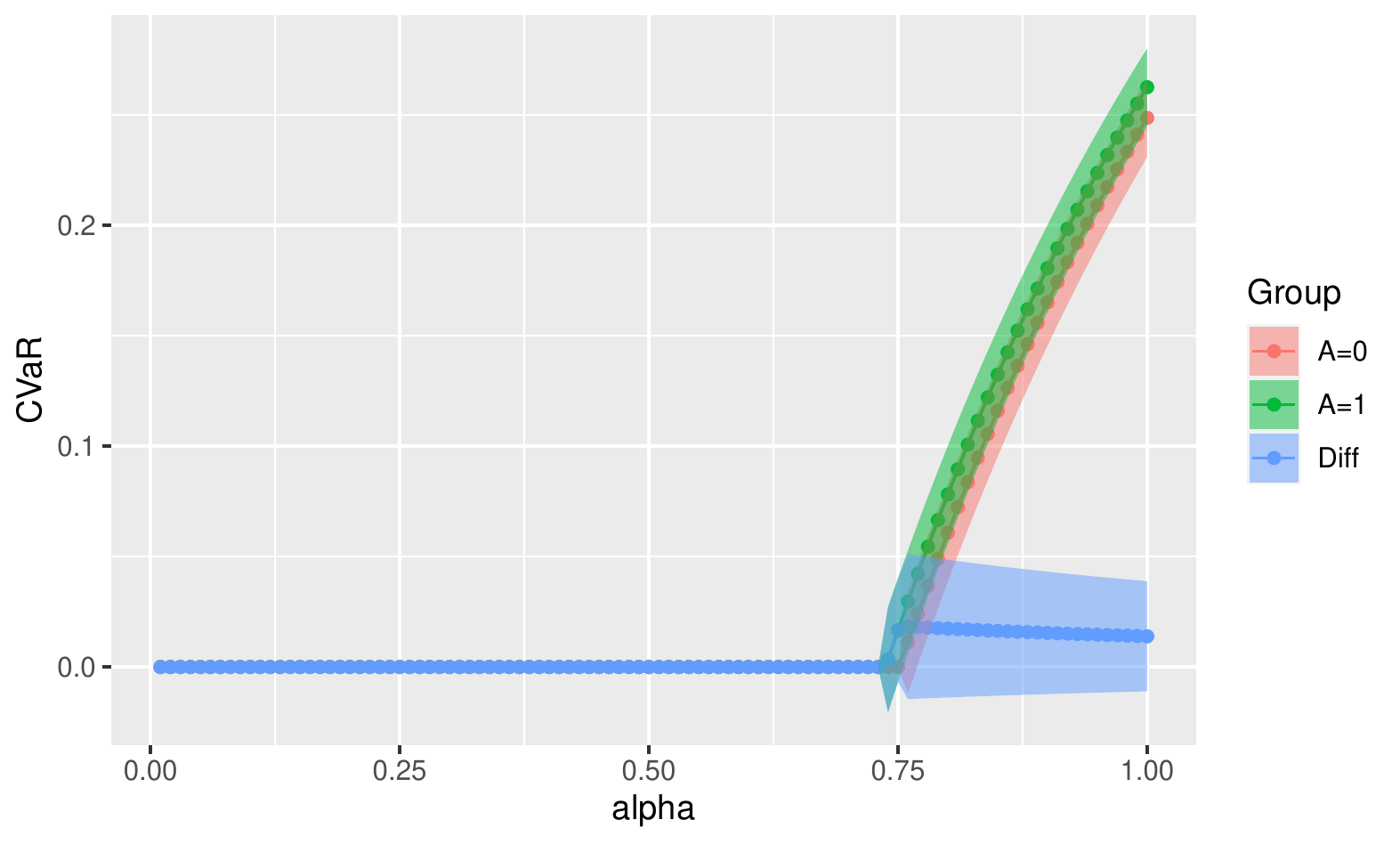}
\caption{The CVaR distributional treatment effect}\label{fig:job_cvar_te}
\end{minipage}
\end{figure}

Here we provide additional plots omitted from \cref{sec:casestudy}.

\Cref{fig:job_cvar_notrearranged} repeats the results presented in \cref{fig:job_cvar} but without applying the rearrangement post-processing suggested in \cref{remark:rearrangement}. Due to estimation error the unprocessed point estimates are not monotonic despite $\cvarat\alpha(\cate)$ being monotonic in $\alpha$. Rearrangement fixes this, making the results easier to interpret, without any loss in precision or inference, and possibly with some gains.

\Cref{fig:job_cvar_te} presents the estimated within-treatment-group CVaR and their difference, being the CVaR distributional treatment effect, that is, $\cvarat\alpha(Y(0))$, $\cvarat\alpha(Y(1))$, and $\cvarat\alpha(Y(1))-\cvarat\alpha(Y(0))$. As can be seen, the result is highly uninformative.

\section{Proofs for Section~\ref{sec:bounds}}

\subsection{Proof of \cref{thm:cvarbound}}

\begin{proof}
By iterated expectations and Jensen's inequality
\begin{align*}
\cvarat\alpha(\ite)
&= \sup_\beta\prns{\beta+\frac1\alpha\E[\E[(\ite-\beta)_-\mid X]]}\\
&\leq \sup_\beta\prns{\beta+\frac1\alpha\E(\cate-\beta)_-}
\\
&= \cvarat\alpha(\cate),
\end{align*}
which yields the first statement.

The second statement follows by setting $\ite=\cate$ and noting integrability ensures the $\cvar$ exists.
\end{proof}

\subsection{Proof of \cref{thm:cvarlb2}}

\begin{proof}
Let $\mathcal P$ be the set of joint distributions $\mathbb P'$ on $(X,\ite)$ having the same $X$-marginal and $\Eb{\ite\mid X}$ as $\mathbb P$ and satisfying $\abs{\cate-\ite}\leq b$ $\mathbb P'$-a.s. Then
\begin{align*}
\cvarat\alpha(\ite)
&\geq
\inf_{\mathbb P'\in\mathcal P}\sup_\beta\prns{\beta+\frac1\alpha\E_{\mathbb P'}[(\ite-\beta)_-]}
\\
&\geq
\sup_\beta\inf_{\mathbb P'\in\mathcal P}\prns{\beta+\frac1\alpha\E_{\mathbb P'}[(\ite-\beta)_-]}
\\
&=
\sup_\beta~\biggl(\beta+\frac1\alpha\E\biggl[
\inf_{\mathbb P'\in\mathcal P}\E_{\mathbb P'}[(\ite-\beta)_-\mid X]
\biggr]\biggr).
\end{align*}
In the first line we used the fact that $\mathbb P\in\mathcal P$. In the second we used weak duality. And, in the third we used iterated expectations and the fact that the restrictions in $\mathcal P$ factor over $\X$.

We proceed to lower bound the inner infimum:
\begin{align*}
\inf_{\mathbb P'\in\mathcal P}\E_{\mathbb P'}[(\ite-\beta)_-\mid X]
&=
\inf_{\substack{\nu~\text{measure on $[-b,b]$}\\\int z d\nu(z)=0\\\int 1 d\nu(z)=1\\\nu\succeq0}}\int(z+\cate-\beta)_-d\nu(z)
\\
&\geq
\sup_{p,\,q\;:\;pz+q\leq(z+\cate-\beta)_-\,\forall z\in[-b,b]}q
\\
&=
\sup_{p}\inf_{z\in[-b,b]}~(z+\cate-\beta)_--pz
\\
&=
\sup_{p}~(((-b+\cate-\beta)_-+bp)\wedge((b+\cate-\beta)_--bp)).
\\
&=
\frac{(-b+\cate-\beta)_-+(b+\cate-\beta)_-}{2}.
\end{align*}
In the first line we wrote the optimization problem as a semi-infinite linear optimization problem. In the second line we used weak duality. In the third line we used that the largest lower bound is the infimum. In the fourth line we used the a concave function is minimized on the boundary. In fifth line we noted that, since the objective in $p$ is convex with two linear pieces and goes to $-\infty$ as either $p\to+\infty$ or $p\to-\infty$, we have that the maximum occurs at the discontinuity point where the two linear pieces meet. The first statement follows by combining.

The second statement is proven by taking 
$\Prb{\ite=\cate-b\mid X}=\frac12$, $\Prb{\ite=\cate+b\mid X}=\frac12$ and noting integrability ensures the $\cvar$ exists.
\end{proof}

\subsection{Proof of \cref{thm:cvarlb1}}

\begin{proof}
By assumption and by change of variables $\gamma=\beta+b$,
\begin{align*}
\cvarat\alpha(\ite)
&\geq \sup_\beta\prns{\beta+\frac1\alpha\E[(\cate-b-\beta)_-]}\\
&= \sup_\gamma\prns{\gamma-b+\frac1\alpha\E(\cate-\gamma)_-}
\\
&= \cvarat\alpha(\cate)-b,
\end{align*}
which yields the first statement.

Fix any $q\in(\alpha,1)$.
Let $\Prb{\ite=\cate-b\mid X}=q$, $\Prb{\ite=\cate+\frac{qb}{1-q}\mid X}=1-q$.
Then, $\Eb{\ite\mid X}=\cate$ and
\begin{align*}
\cvarat\alpha(\ite)&=
\sup_\beta\prns{\beta+\frac1\alpha\Eb{q(\cate-b-\beta)_-+(1-q)\prns{\cate+\frac{qb}{1-q}-\beta}_-}}\\
&\leq
\sup_\beta\prns{\beta+\frac1\alpha\E[q(\cate-b-\beta)_-]}\\
&=
\cvarat{\alpha/q}(\cate)-b.
\end{align*}
Integrability ensures all these expectations exist.
Note that $\cvarat{\alpha/q}(\cate)$ is concave in $q$ as it is the supremum of affine functions in $q$. Therefore, it must be continuous in $q$. The statement is concluded by taking $q\to1$.
\end{proof}

\subsection{Proof of \cref{thm:cvarlb3}}

\begin{proof}
For any given distribution on $\X$ and a square-integrable function $\catef:\X\to\Rl$,
let $\mathcal P$ be the set of joint distributions $\mathbb P'$ on $(X,\ite)$ having the given $X$-marginal and $\Eb{\ite\mid X}=\cate$ as $\mathbb P$ and satisfying $\op{Var}(\ite\mid X)\leq \bar\sigma^2(X)$ $\mathbb P'$-a.s. We will proceed to prove both statements by directly evaluating the program
\begin{equation}\label{eq:cvarminimax}
\inf_{\mathbb P'\in\mathcal P}\sup_\beta\prns{\beta+\frac1\alpha\E_{\mathbb P'}[(\ite-\beta)_-]},
\end{equation}
and showing it is equal to the right-hand side of \cref{eq:cvarlb3}.

Note that the objective of \cref{eq:cvarminimax} is linear in $\mathbb P'$ and concave in $\beta$.
Moreover, since every $\mathbb P'\in\mathcal P$ has $\E_{\mathbb P'}\ite^2\leq (\E\cate)^2+\op{Var}(\cate)+\E\bar\sigma^2(X)<\infty$ by assumption of integrability, we have that $\mathcal P'$ is compact.
Therefore, by Sion's minimax theorem, we have
\begin{align*}
\inf_{\mathbb P'\in\mathcal P}\sup_\beta\prns{\beta+\frac1\alpha\E_{\mathbb P'}[(\ite-\beta)_-]}&=\sup_\beta\inf_{\mathbb P'\in\mathcal P}\prns{\beta+\frac1\alpha\E_{\mathbb P'}[(\ite-\beta)_-]}
\\
&=\sup_\beta~\biggl(\beta+\frac1\alpha\E\biggl[
\inf_{\mathbb P'\in\mathcal P}\E_{\mathbb P'}[(\ite-\beta)_-\mid X]
\biggr]\biggr),
\end{align*}
where in the second equality we used iterated expectations and the fact that the restrictions in $\mathcal P$ factor over $\X$.

We proceed to compute the inner infimum. Fix $\beta$ and $X$, and set $m=\beta-\tau(X),\,s=\bar\sigma^2(X)$. The inner infimum is equal to the following semi-infinite optimization problem over a signed measure as the decision variable with a linear objective and two linear equality constraints and one convex-quadratic constraint:
\begin{align*}
\ostar~=~\inf_{\nu}\quad&\int(z-m)_-d\nu(z)\\
\text{s.t.}\quad&\int zd\nu(z)=0\\
&\int 1d\nu(z)=1\\
&\int z^2d\nu(z)\leq s\\
&\nu\succeq0.
\end{align*}

Using $\nu$ being a Dirac at zero as a Slater point, strong duality for semi-infinite optimization gives an equivalent optimization problem in three scalar decision variables and a continuum of constraints:
\begin{align*}
\ostar~=~\sup_{p,q,r}\quad&q-sr\\
\text{s.t.}\quad&pz+q-z^2\leq(z-m)_-~~\forall z\\
&r\geq0.
\end{align*}

Using that the largest lower bound is the infimum and that if we choose $r=0$ in the outer supremum then $z\to\infty$ achieves $\infty$ in the inner infimum, we obtain
$$
\ostar~=~\sup_{p\in\Rl,\,r\in\Rl_{++}}\inf_z\quad(z-m)_--pz+(z^2-s)r.$$

We now compute the inner infimum. The objective is the minimum of two convex quadratics. Therefore, the infimum is equal to the minimum of the infimum of each quadratic by itself. We obtain
$$
\ostar~=~\sup_{p\in\Rl,\,r\in\Rl_{++}}\quad\prns{\frac{-p^2+2p-4r(m+rs)-1}{4r}}\wedge\prns{\frac{-p^2-rs}{4r}}.
$$
Fix $r\in\Rl_{++}$ and consider the supremum in $p$ alone. The objective is the minimum of two strictly concave quadratics with the \emph{same} quadratic part and \emph{different} linear parts. Therefore, the maximum occurs at the \emph{single} discontinuity point where the two quadratics meet. We obtain
$$
\ostar~=~\sup_{r\in\Rl_{++}}\quad-\frac{(1+4mr)^2}{16r}-sr.
$$

This objective is differentiable, convex for $r$ positive, and approaches $-\infty$ both as $r\downarrow0$ and as $r\uparrow\infty$. Hence, the maximum occurs at a critical point in the positive half-line. This critical point is at $r=1/(4\sqrt{m^2+s})$.
We obtain
$$
\ostar~=~\frac12(-m-\sqrt{m^2+s}),
$$
completing the proof.
\end{proof}

\subsection{Proof of \cref{cor:cvarlb4}}

\begin{proof}
The first inequality follows from \cref{thm:cvarbound}. The second inequality follows by applying \cref{thm:cvarlb3} with $\bar\sigma^2(X)=\op{Var}(\ite\mid X)$ and observing that, by triangle inequality,
\begin{align*}
\Eb{\cate-\beta-\sqrt{(\cate-\beta)^2+\bar\sigma^2(X)}}
&\geq
\Eb{\cate-\beta-\abs{\cate-\beta}-\bar\sigma(X)}
=
\Eb{2(\cate-\beta)_--\bar\sigma(X)}.
\end{align*}
The third inequality follows by Cauchy-Schwarz, and the fourth by Jensen's and iterated expectation.
\end{proof}

\section{Proofs for Section~\ref{sec:localrobust}}

\subsection{Preliminary lemma}

\begin{lemma}\label{lemma:EphiErr}
Suppose \cref{asm:regularity} holds.
Then, there exists constants $c_1>0,c_2>0,c_3>0$ such that
for any $\alpha\in(0,1]$ and any $\tprop,\tmu,\tcate,\tbeta$ with 
$\bar e\leq\tprop\leq 1-\bar e$, 
$\|\tmu\|_\infty\leq B$,
$\|\tau-\tcate\|_{\infty}\leq c_1$, 
and $\fabs{\beta^*-\tbeta}\leq c_1$,
we have
\begin{align*}
\abs{\E[\phi(X,A,\Yobs;\tprop,\tmu,\tcate,\tbeta)]
-
\E[\phi(X,A,\Yobs; e,\mu,\tau,\beta^*)]}
&\leq \frac{c_2}{\alpha}
\prns{\magd{e-\tprop}_2\magd{\mu-\tmu}_2+\magd{\tau-\tcate}_\infty^2+(\beta^*-\tbeta)^2},
\\
\magd{\phi(X,A,\Yobs;\tprop,\tmu,\tcate,\tbeta)
-
\phi(X,A,\Yobs; e,\mu,\tau,\beta^*)}_2&\leq\frac{c_3}{\alpha}
\prns{\magd{e-\tprop}_2+\magd{\mu-\tmu}_2+\magd{\tau-\tcate}_\infty+\fabs{\beta^*-\tbeta}}
.
\end{align*}
\end{lemma}

\begin{proof}
First, we compute:
\begin{align*}
\E\phi(X,A,\Yobs;\tprop,\tmu,\tcate,\tbeta)=
\tbeta+\frac1\alpha\E\biggl[
&
\indic{\tcate(X)\leq\tbeta}
\\&\times
\prns{
\tmu(X,1)-\tmu(X,0)+\frac{e(X)}{\tprop(X)}(\mu(X,1)-\tmu(X,1))-\frac{1-e(X)}{1-\tprop(X)}(\mu(X,0)-\tmu(X,0))-\tbeta
}
\biggr].
\end{align*}

We proceed to show the first inequality by bounding each of the following:
\begin{align}
\label{eq:orthb 1}&\fabs{\E\phi(X,A,\Yobs;\tprop,\tmu,\tcate,\tbeta)-\E\phi(X,A,\Yobs;e,\tmu,\tcate,\tbeta)},\\
\label{eq:orthb 1b}&\fabs{\E\phi(X,A,\Yobs; e,\tmu,\tcate,\tbeta)-\E\phi(X,A,\Yobs;e, \mu,\tcate,\tbeta)},\\
\label{eq:orthb 2}&\fabs{\E\phi(X,A,\Yobs; e,\mu,\tcate,\tbeta)-\E\phi(X,A,\Yobs;e, \mu,\tau,\tbeta)},\\
\label{eq:orthb 3}&\fabs{\E\phi(X,A,\Yobs; e, \mu,\tau,\tbeta)-\E\phi(X,A,\Yobs;e, \mu,\tau, \beta^*)}.
\end{align}

We begin with \cref{eq:orthb 1}. We have
\begin{align*}
\fabs{\E\phi(X,A,\Yobs;\tprop,\tmu,\tcate,\tbeta)-\E\phi(X,A,\Yobs;e,\tmu,\tcate,\tbeta)}\leq~&\frac{1}{\alpha}\Eb{\indic{\tcate(X)\leq\tbeta}\frac1{\tprop(X)}\abs{e(X)-\tprop(X)}\abs{\mu(X,1)-\tmu(X,1)}}
\\&+\frac{1}{\alpha}\Eb{\indic{\tau(X)\leq\tbeta}\frac1{1-\tprop(X)}\abs{e(X)-\tprop(X)}\abs{\mu(X,0)-\tmu(X,0)}}\\
\leq&\frac{1}{\alpha\bar e}\magd{e-\tprop}_2\prns{\magd{\mu(\cdot,1)-\tmu(\cdot,1)}_2+\magd{\mu(\cdot,0)-\tmu(\cdot,0)}_2}
.
\end{align*}

Next, we observe that \cref{eq:orthb 1b} is exactly $0$.

Next, we tackle \cref{eq:orthb 2}.
By \cref{asm:regularity}, there exists $c>0$ such that $\cate-\beta^*$ has a density on $(-c,c)$ bounded by $F'_{\cate}(F_{\cate}^{-1}(\alpha))+1$.
Therefore, provided that $\fabs{\tbeta-\beta^*}\leq c/3,\,\fmagd{\tcate(X)-\tau(X)}_\infty\leq c/3$,
we have
\begin{align*}\notag
&\abs{\E\phi(X,A,\Yobs; e,\mu,\tcate,\tbeta)-\E\phi(X,A,\Yobs;e, \mu,\tau,\tbeta)}
\\&\qquad=
\frac1{\alpha}\abs{\Eb{
\fprns{\findic{\cate-\beta^*\leq \tbeta-\beta^*+\tau(X)-\tcate(X)}-\findic{\cate-\beta^*\leq \tbeta-\beta^*}}
\prns{
\cate-\beta^*
}
}}
\\&\qquad\leq
\frac1\alpha\Eb{\fabs{\cate-\beta^*}\,\findic{\fabs{\cate-\beta^*}\leq \fabs{\tbeta-\beta^*}+\fabs{\tcate(X)-\tau(X)}}}
\\&\qquad\leq
\frac1\alpha\Eb{\fabs{\cate-\beta^*}\,\findic{\fabs{\cate-\beta^*}\leq \fabs{\tbeta-\beta^*}+\fmagd{\tcate(X)-\tau(X)}_\infty}}
\\&\qquad\leq
\frac1\alpha\fprns{F'_{\cate}(F_{\cate}^{-1}(\alpha))+1}\prns{\fabs{\tbeta-\beta^*}+\fmagd{\tcate(X)-\tau(X)}_\infty}^2
.
\end{align*}

We now tackle \cref{eq:orthb 3}.
Let
\begin{equation}\label{eq:orth 3 bound}
f(\beta)=\E[\phi(X,A,\Yobs;e,\mu,\tau,\beta))]
=
\beta+\frac1\alpha\Eb{
\prns{
\cate-\beta
}_-
}.
\end{equation}
By assumption $f'(\beta^*)=0$ and $\abs{f''(\beta)}\leq \frac1\alpha\fprns{F'_{\cate}(F_{\cate}^{-1}(\alpha))+1}$ for $\beta\in(\beta^*-c,\beta^*+c)$. Therefore, provided $\fabs{\tbeta-\beta^*}\leq c/3$, Taylor's theorem yields that \cref{eq:orthb 3} is bounded by $\frac1{2\alpha}\fprns{F'_{\cate}(F_{\cate}^{-1}(\alpha))+1}\fprns{\tbeta-\beta^*}^2$.

We proceed to show the second inequality by bounding each of the following:
\begin{align}
\label{eq:orthc 1}&\fmagd{\phi(X,A,\Yobs;\tprop,\tmu,\tcate,\tbeta)-\phi(X,A,\Yobs;e,\tmu,\tcate,\tbeta)}_2,\\
\label{eq:orthc 1b}&\fmagd{\phi(X,A,\Yobs; e,\tmu,\tcate,\tbeta)-\phi(X,A,\Yobs;e, \mu,\tcate,\tbeta)}_2,\\
\label{eq:orthc 2}&\fmagd{\phi(X,A,\Yobs; e,\mu,\tcate,\tbeta)-\phi(X,A,\Yobs;e, \mu,\tcate,\beta^*)}_2,\\
\label{eq:orthc 3}&\fmagd{\phi(X,A,\Yobs; e, \mu,\tcate,\beta)-\phi(X,A,\Yobs;e, \mu,\tau, \beta^*)}_2.
\end{align}

By writing $\phi(X,A,\Yobs;\tprop,\tmu,\tcate,\tbeta)-\phi(X,A,\Yobs;e,\tmu,\tcate,\tbeta)=\frac1\alpha\indic{\tcate(X)\leq\tbeta}(Y-\tmu(X,A))(A(\tprop^{-1}(X)-e^{-1}(X))-(1-A)((1-\tprop(X))^{-1}-(1-e(X))^{-1})$ and noting that $\|A(\tprop^{-1}(X)-e^{-1}(X))\|_2\leq\frac1{\bar e^{3/2}}\|\tprop-e\|_2$, we see that \cref{eq:orthc 1} is bounded by $\frac{4B}{\alpha\bar e^{3/2}}\|\tprop-e\|_2$.

By writing $\phi(X,A,\Yobs; e,\tmu,\tcate,\tbeta)-\phi(X,A,\Yobs;e, \mu,\tcate,\tbeta)=\frac1\alpha\indic{\tcate(X)\leq\tbeta}
\prns{1-\frac{A}{e(X)}}(\tmu(X,1)-\mu(X,1))
+
\frac1\alpha\indic{\tcate(X)\leq\tbeta}
\prns{\frac{1-A}{1-e(X)}-1}(\tmu(X,0)-\mu(X,0))
$, we see that \cref{eq:orthc 1b} is bounded by $\frac{1}{\alpha\bar e}\prns{\|\tmu(X,0)-\mu(X,0)\|_2+\|\tmu(X,1)-\mu(X,1)\|_2}$.

To bound \cref{eq:orthc 2} let us first write $\phi(X,A,\Yobs; e,\mu,\tcate,\tbeta)-\phi(X,A,\Yobs;e, \mu,\tcate,\beta^*)=(\tbeta-\beta^*)(1-\frac1\alpha\findic{\tcate(X)\leq\tbeta})+\frac1\alpha(\findic{\tcate(X)\leq\tbeta}-\findic{\tcate(X)\leq\beta^*})(\mu(X,1)-\mu(X,0)+\frac{A-e(X)}{e(X)(1-e(X))}(Y-\mu(X,A))-\beta^*)$. By \cref{asm:regularity}, we have $\abs{\beta^*}\leq 2B$. Therefore, \cref{eq:orthc 2} is bounded by $\frac1\alpha\fabs{\tbeta-\beta^*}+\frac{6B}{\alpha\bar e}\fmagd{\findic{\tcate(X)\leq\tbeta}-\findic{\tcate(X)\leq\beta^*}}_2$. And, provided $\fabs{\tbeta-\beta^*}\leq c/3,\,\fmagd{\tcate(X)-\tau(X)}_\infty\leq c/3$, we have
 $\fmagd{\findic{\tcate(X)\leq\tbeta}-\findic{\tcate(X)\leq\beta^*}}_2\leq \Prb{\abs{\cate-\beta^*}\leq \fabs{\tbeta-\beta^*}+\fmagd{\tcate(X)-\tau(X)}_\infty}\leq 2\fprns{F'_{\cate}(F_{\cate}^{-1}(\alpha))+1}\fprns{\fabs{\tbeta-\beta^*}+\fmagd{\tcate(X)-\tau(X)}_\infty}$

Finally, \cref{eq:orthc 2} is bounded by $\frac{6B}{\alpha\bar e}\fmagd{\findic{\tcate(X)\leq\beta^*}-\findic{\tau(X)\leq\beta^*}}_2$.
And, provided $\fmagd{\tcate(X)-\tau(X)}_\infty\leq c/3$, we have
$\fmagd{\findic{\tcate(X)\leq\beta^*}-\findic{\tau(X)\leq\beta^*}}_2\leq \Prb{\abs{\cate-\beta^*}\leq \fmagd{\tcate(X)-\tau(X)}_\infty}\leq 2\fprns{F'_{\cate}(F_{\cate}^{-1}(\alpha))+1}\fmagd{\tcate(X)-\tau(X)}_\infty$.
\end{proof}


\subsection{Proof of \cref{lemma:betalemma}}

\begin{proof}
Set
$\mathcal I_{-k}=\braces{i\not\equiv k-1~\text{(mod $K$)}}$.
For any function $f(x)$, let us denote
$Q_\alpha(f)=
\inf\fbraces{\beta:\Eb{\indic{f(X)\leq\beta}-\alpha}\geq0}
$ 
and
$\hat Q^{(k)}_\alpha(f)=\inf\fbraces{\beta:\sum_{i\in\mathcal I_{-k}}\prns{\indic{f(X_i)\leq\beta}-\alpha}\geq0}$
so that $\hat\beta^{(k)}=\hat Q^{(k)}_\alpha(\hat\catef^{(k)})$
and $\beta^*=F_{\cate}^{-1}(\alpha)=Q_\alpha(\catef)$.

First we consider the case $r=\infty$. We have $\abs{\hat Q^{(k)}_\alpha(\hat\catef^{(k)})-\hat Q_\alpha(\catef)}\leq \sup_{i\in\mathcal I_{-k}}\abs{\hat\catef^{(k)}(X_i)-\catef(X_i)}=O_p(\fmagd{\hat\catef^{(k)}-\tau}_\infty)$.
By \cref{asm:regularity} and delta method, $\abs{\hat Q^{(k)}_{\alpha}(\catef)-Q_{\alpha}(\catef)}=O_p(n^{-1/2})$, giving the result for $r=\infty$.

Consider now $r<\infty$.
Let $\delta=\magd{\catef-\hat\catef}_r^{\frac{r}{r+1}}$. By union bound with respect to the empirical distribution (conditioning on the data),
$$
\hat Q^{(k)}_\alpha(\hat\catef^{(k)})\leq 
\hat Q^{(k)}_{\alpha+\delta}(\catef)
+
\hat Q^{(k)}_{1-\delta}(\hat\catef^{(k)}-\catef).
$$
By \cref{asm:regularity} and delta method, $\hat Q^{(k)}_{\alpha+\delta}(\catef)\leq Q_{\alpha}(\catef)+O_p(\delta)+O_p(n^{-1/2})$.
By Markov's inequality with respect to the empirical distribution (conditioning on the data),
$$
\hat Q^{(k)}_{1-\delta}(\hat\catef^{(k)}-\catef)\leq 
\delta^{-1/r}\prns{\frac1{\abs{\mathcal I_{-k}}}\sum_{i\in\mathcal I_{-k}}\abs{\hat\catef^{(k)}(X_i)-\catef(X_i)}^r}^{1/r}
= O_p(\delta^{-1/r}\magd{\catef-\hat\catef}_r)=O_p(\delta).
$$
A wholly symmetric argument for a lower bound gives the conclusion.
\end{proof}

\subsection{Proof of Theorem~\ref{thm:asympnormal}}

\begin{proof}
Let
$\mathcal I_k=\braces{i\equiv k-1~\text{(mod $K$)}}$,
$\mathcal I_{-k}=\braces{i\not\equiv k-1~\text{(mod $K$)}}$,
$\hat\E_{k}f(X,A,Y)=\frac{1}{\abs{\mathcal I_k}}\sum_{i\in\mathcal I_k}f(X_i,A_i,Y_i)$, and $\E_{\mid-k}f(X,A,Y)=\E[f(X,A,Y)\mid \{(X_i,A_i,Y_i):i\in \mathcal I_{-k}\}]$.
\begin{align}
\notag&\hat\E_{k}\phi(X,A,\Yobs;\hat e^{(k)},\hat\mu^{(k)},\hat\tau^{(k)},\hat\beta^{(k)})-\hat\E_{k}\phi(X,A,\Yobs;e,\mu,\tau,\beta^*)\\
&=\label{eq:asymp a}\E_{\mid-k}\phi(X,A,\Yobs;\hat e^{(k)},\hat\mu^{(k)},\hat\tau^{(k)},\hat\beta^{(k)})-\E_{\mid-k}\phi(X,A,\Yobs;e,\mu,\tau,\beta^*)\\&\label{eq:asymp b}\phantom{=}+(\hat\E_{k}-\E_{\mid-k})(\phi(X,A,\Yobs;\hat e^{(k)},\hat\mu^{(k)},\hat\tau^{(k)},\hat\beta^{(k)})-\phi(X,A,\Yobs;e,\mu,\tau,\beta^*)).
\end{align}
We proceed to show that each of \cref{eq:asymp a,eq:asymp b} are $o_p(1/\sqrt{n})$.

By \cref{lemma:EphiErr}, we have that \cref{eq:asymp a} is
$$O_p(\fmagd{e-\hat e^{(k)}}_2\fmagd{\mu-\hat \mu^{(k)}}_2
+\fmagd{\tau-\hat\tau^{(k)}}_\infty^2+(\hat\beta^{(k)}-\beta^*)^2
).$$
So, by our nuisance-estimation assumptions and \cref{lemma:betalemma}, \cref{eq:asymp a} is $o_p(1/\sqrt{n})$.

By Chebyshev's inequality conditioned on $\mathcal I_{-k}$, we obtain that \cref{eq:asymp b} is
$$
O_p(\abs{\mathcal I_k}^{-1/2}\fmagd{\phi(X,A,\Yobs;\hat e^{(k)},\hat\mu^{(k)},\hat\tau^{(k)},\hat\beta^{(k)})-\phi(X,A,\Yobs;e,\mu,\tau,\beta^*)}_2).
$$
By our nuisance-estimation assumptions, \cref{lemma:betalemma}, and \cref{lemma:EphiErr}, we have that $\fmagd{\phi(X,A,\Yobs;\hat e^{(k)},\hat\mu^{(k)},\hat\tau^{(k)},\hat\beta^{(k)})-\phi(X,A,\Yobs;e,\mu,\tau,\beta^*)}_2=o_p(1)$.
Thus, \cref{eq:asymp b} is $o_p(1/\sqrt{n})$.

The rest is concluded by noting that $\frac1n\sum_{i=1}^n\phi(X_i,A_i,Y_i;e,\mu,\tau,\beta^*)=\frac1K\sum_{k=1}^K\frac{\abs{\mathcal I_k}}{n/K}\hat\E_k \phi(X,A,Y;e,\mu,\tau,\beta^*)$ and by the central limit theorem.
\end{proof}

\section{Proofs for Section~\ref{sec:caterobust}}

\subsection{Preliminary lemma}

\begin{lemma}\label{lemma:EphiErr2}
Fix any $\sprop,\smu,\scate$ with
$\bar e\leq \sprop\leq 1-\bar e$, $\fmagd{\smu}_\infty\leq B$, $\fmagd{\scate}_\infty\leq 2B$.
Suppose that either $\sprop=e$ or $\smu=\mu$.
Suppose \cref{asm:regularity} holds with $\tau$ replaced by $\scate$, and set $\sbeta=F^{-1}_{\scate(X)}(\alpha)$.
Set $\kappa=1$ if $\scate=\tau$ and otherwise set $\kappa=0$.
Then, there exists constants $c_1>0,c_2>0,c_3>0$ such that
for any $\alpha\in(0,1]$ and any $\tprop,\tmu,\tcate,\tbeta$ with 
$\bar e\leq\tprop\leq 1-\bar e$, 
$\|\tmu\|_\infty\leq B$,
$\|\scate-\tcate\|_{\infty}\leq c_1$, 
and $\fabs{\sbeta-\tbeta}\leq c_1$,
we have
\begin{align*}
\abs{\E[\phi(X,A,\Yobs;\tprop,\tmu,\tcate,\tbeta)]
-
\E[\phi(X,A,\Yobs; \sprop,\smu,\scate,\sbeta)]}
&\leq \frac{c_2}{\alpha}
\bigl(
\fmagd{\sprop-\tprop}_2\fmagd{\smu-\tmu}_2
+\fmagd{\sprop-\tprop}_2\fmagd{\smu-\mu}_2
\\&\phantom{\leq \frac{c_2}{\alpha}\bigl(}+\fmagd{\sprop-e}_2\fmagd{\smu-\tmu}_2+\fmagd{\scate-\tcate}^{1+\kappa}_\infty+\fabs{\sbeta-\tbeta}^{1+\kappa}\bigr),
\\
\magd{\phi(X,A,\Yobs;\tprop,\tmu,\tcate,\tbeta)
-
\phi(X,A,\Yobs; \sprop,\smu,\scate,\sbeta)}_2&\leq\frac{c_3}{\alpha}
\prns{\fmagd{\sprop-\tprop}_2+\fmagd{\smu-\tmu}_2+\fmagd{\scate-\tcate}_\infty+\fabs{\sbeta-\tbeta}}
.
\end{align*}
\end{lemma}

\begin{proof}
We will proceed to bound each of
\begin{align}
\label{eq:orth3 1}&\abs{\E\phi(X,A,\Yobs;\tprop,\tmu,\tcate,\tbeta)-\E\phi(X,A,\Yobs;\sprop,\smu,\tcate,\tbeta)},\\
\label{eq:orth3 2}&\abs{\E\phi(X,A,\Yobs; \sprop,\smu,\tcate,\tbeta)-\E\phi(X,A,\Yobs;\sprop, \smu,\scate,\tbeta)},\\
\label{eq:orth3 3}&\abs{\E\phi(X,A,\Yobs; \sprop, \smu,\scate,\tbeta)-\E\phi(X,A,\Yobs;\sprop, \smu, \scate, \sbeta)}.
\end{align}

First, consider the case that $\sprop=e$. We will bound \cref{eq:orth3 1} by bounding each of
\begin{align}
\label{eq:orth3 1 1a}&\abs{\E\phi(X,A,\Yobs;\tprop,\tmu,\tcate,\tbeta)-\E\phi(X,A,\Yobs;e,\tmu,\tcate,\tbeta)},\\
\label{eq:orth3 1 1b}&\abs{\E\phi(X,A,\Yobs;e,\tmu,\tcate,\tbeta)-\E\phi(X,A,\Yobs;e,\smu,\tcate,\tbeta)}.
\end{align}
We begin with \cref{eq:orth3 1 1a}. We have
\begin{align*}
\E\phi(X,A,\Yobs;\tprop,\tmu,\tcate,\tbeta)-\E\phi(X,A,\Yobs;e,\tmu,\tcate,\tbeta)=~&\frac{1}{\alpha}\Eb{\indic{\tcate(X)\leq\tbeta}\frac1{\tprop(X)}(e(X)-\tprop(X))(\mu(X,1)-\tmu(X,1))}
\\&+\frac{1}{\alpha}\Eb{\indic{\tcate(X)\leq\tbeta}\frac1{1-\tprop(X)}(e(X)-\tprop(X))(\mu(X,0)-\tmu(X,0))}\\
\leq&\frac{1}{\alpha\bar e}\magd{e-\tprop}_2\prns{\fmagd{\smu(\cdot,1)-\tmu(\cdot,1)}_2+\fmagd{\smu(\cdot,0)-\tmu(\cdot,0)}_2}\\
&+\frac{1}{\alpha\bar e}\magd{e-\tprop}_2\prns{\fmagd{\mu(\cdot,1)- \smu(\cdot,1)}_2+\fmagd{\mu(\cdot,0)- \smu(\cdot,0)}_2}
.
\end{align*}
Next, we observe that \cref{eq:orth3 1 1b} is exactly 0.

Second, we consider the case that $\smu=\mu$.
We will bound \cref{eq:orth3 1} by bounding each of
\begin{align}
\label{eq:orth3 1 2a}&\abs{\E\phi(X,A,\Yobs;\tprop,\tmu,\tcate,\tbeta)-\E\phi(X,A,\Yobs;\tprop,\mu,\tcate,\tbeta)},\\
\label{eq:orth3 1 2b}&\abs{\E\phi(X,A,\Yobs;\tprop,\mu,\tcate,\tbeta)-\E\phi(X,A,\Yobs;\sprop,\mu,\tcate,\tbeta)}.
\end{align}
We begin with \cref{eq:orth3 1 2a}. We have
\begin{align*}
\E\phi(X,A,\Yobs;\tprop,\tmu,\tcate,\tbeta)-\E\phi(X,A,\Yobs;\tprop,\mu,\tcate,\tbeta)=~&\frac{1}{\alpha}\Eb{\indic{\tcate(X)\leq\tbeta}\frac1{\tprop(X)}(\tprop(X)-e(X))(\tmu(X,1)-\mu(X,1))}
\\&+\frac{1}{\alpha}\Eb{\indic{\tcate(X)\leq\tbeta}\frac1{1-\tprop(X)}(\tprop(X)-e(X))(\tmu(X,0)-\mu(X,0))}\\
\leq&\frac{1}{\alpha\bar e}\fmagd{\sprop-\tprop}_2\prns{\fmagd{\mu(\cdot,0)-\tmu(\cdot,0)}_2+\fmagd{\mu(\cdot,1)-\tmu(\cdot,1)}_2}\\
&+\frac{1}{\alpha\bar e}\fmagd{\sprop-e}_2\prns{\fmagd{\mu(\cdot,0)-\tmu(\cdot,0)}_2+\fmagd{\mu(\cdot,1)-\tmu(\cdot,1)}_2}
.
\end{align*}
Next, we observe that \cref{eq:orth3 1 2b} is exactly 0.

If $\kappa=1$ then the rest is argued as in \cref{lemma:EphiErr}.

Next, we tackle \cref{eq:orth3 2,eq:orth3 3} supposing $\kappa=0$.
Because either $\sprop=e$ or $\smu=\mu$, we have
\begin{align*}
&\abs{\E\phi(X,A,\Yobs; \sprop,\smu,\tcate,\tbeta)-\E\phi(X,A,\Yobs;\sprop, \smu,\scate,\tbeta)}
\\&\qquad=
\frac1{\alpha}\abs{\Eb{
\fprns{\findic{\scate(X)-\sbeta\leq \tbeta-\sbeta+\scate(X)-\tcate(X)}-\findic{\scate(X)-\sbeta\leq \tbeta-\sbeta}}
\prns{
\cate-\sbeta
}
}}.
\end{align*}
By assumption, there exists $c>0$ such that $\scate(X)-\beta^*$ has a density on $(-c,c)$ bounded by $F'_{\scate}(F_{\scate(X)}^{-1}(\alpha))+1$. Moreover, since $\fmagd{\scate(X)}_\infty\leq 2M$, we have $\fabs{\sbeta}\leq 2M$.
Therefore, provided that $\fabs{\tbeta-\sbeta}\leq c/3,\,\fmagd{\tcate(X)-\scate(X)}_\infty\leq c/3$,
we have that \cref{eq:orth3 2} is bounded by
\begin{align*}
\frac{4M}{\alpha}\prns{F'_{\scate(X)}(F_{\scate(X)}^{-1}(\alpha))+1}\prns{\fabs{\tbeta-\sbeta}+\fmagd{\tcate(X)-\scate(X)}_\infty}.
\end{align*}

Finally, provided that $\fabs{\tbeta-\sbeta}\leq c/3$, \cref{eq:orth3 3} is bounded by
\begin{align*}
&\fabs{\tbeta-\sbeta}+
\frac1\alpha\Eb{
\abs{\findic{\scate(X)\leq\tbeta}-\findic{\scate(X)\leq\sbeta}}\abs{\cate-\sbeta}}
+
\frac1\alpha\Eb{
\findic{\scate(X)\leq\tbeta}\fabs{\tbeta
-
\sbeta}
}
\\
&\leq \prns{1+\frac1\alpha+\frac{8M}{\alpha}\prns{F'_{\scate(X)}(F_{\scate(X)}^{-1}(\alpha))+1}}\fabs{\tbeta-\sbeta}
\end{align*}

The second inequality is proven the same as in \cref{lemma:EphiErr}.
\end{proof}

\subsection{Proof of \cref{lemma:popdoublevalid}}

\begin{proof}
By assumption of continuity on $\scate(X)$, we have that $\fPrb{\scate(X)\leq\sbeta}=\alpha$ and that the right-hand side of \cref{eq:popvalid} is equal to $\Eb{\frac1\alpha\findic{\scate(X)\leq\sbeta}\cate}$. On the other hand, by the dual formulation of CVaR \citep{rockafellar2000optimization}, we have
$$
\cvarat\alpha(\cate)=\inf_{0\leq Z\leq\frac1\alpha,\,\E Z=1}\E[Z\cate].
$$
Since $Z=\frac1\alpha\findic{\scate(X)\leq\sbeta}$ is feasible, the conclusion follows.
\end{proof}

\subsection{Proof of Theorem~\ref{thm:doublyrobust}}

\begin{proof}
The proof proceeds as in \cref{thm:asympnormal} but using \cref{lemma:EphiErr2} and noting that 
$\E\phi(X,A,Y;\sprop,\smu,\tau,\beta^*)=\Psi$ provided that either $\sprop=e$ or $\smu=\mu$.
\end{proof}

\subsection{Proof of Theorem~\ref{thm:doublyvalid}}

\begin{proof}
The proof proceeds as in \cref{thm:asympnormal}  but using \cref{lemma:EphiErr2} and noting that 
\cref{lemma:betalemma} applies analogously for $\scate$, that we have
$\E\phi(X,A,Y;\sprop,\smu,\scate,\sbeta)=\sbeta+\frac1\alpha\Efb{\findic{\scate(X)\leq\sbeta}(\cate-\sbeta)}$ provided that either $\sprop=e$ or $\smu=\mu$, and finally that this upper bounds $\Psi$ by \cref{lemma:popdoublevalid}.
\end{proof}

\end{document}